\documentclass{svproc}
\usepackage{url}

\newtheorem{comment}{Comment}

\def\boldhead#1:{\par\vskip 7pt\noindent{\bf #1:}\hskip 10pt}
\def\ithead#1:{\par\vskip 7pt\noindent{\it #1:}\hskip 10pt}
\def\inline#1:{\par\vskip 7pt\noindent{\bf #1:}\hskip 10pt}
\long\def\comment #1\commentend{}
\long\def\commfull #1\commend{#1}
\long\def\commabs #1\commenda{}
\long\def\commtim #1\commendt{#1}
\long\def\commb #1\commbend{}
\def\blackslug{\hbox{\hskip 1pt \vrule width 4pt height 8pt
		depth 1.5pt \hskip 1pt}}
\def\QED{\quad\blackslug\lower 8.5pt\null\par}
\def\Qed{\QED}
\def\inQED{~~~~~\quad\blackslug\lower 8.5pt\null}
\def\Proof{\noindent{\bf Proof:~}}
\def\ProofOf#1#2{\noindent{\bf Proof of #1 \ref{#2}:~}}

\long\def\PPP#1{\noindent{\bf Proof:}{ #1}{\quad\blackslug\lower 8.5pt\null}}

\long\def\denspar #1\densend
{#1}

\newcommand{\Set}[1]{\left\{ #1 \right\}}

\newcommand{\ceil}[1]{\left\lceil #1 \right\rceil}
\newcommand{\floor}[1]{\left\lfloor #1 \right\rfloor}
\newcommand{\ignore}[1]{}

\def\DEF{\stackrel{\rm def}{=}}

\def\cC{{\cal C}}

\def\cF{{\cal F}}

\def\cP{{\cal P}}

\def\cY{{\cal Y}}

\newif\ifnotesw\noteswtrue
{\ifnotesw\marginpar[\hfill\(\top\)]{\(\top\)}\fi}%
{\ifnotesw\marginpar[\hfill\(\bot\)]{\(\bot\)}\fi}

\newcommand{\mnote}[1]%
{\ifnotesw\marginpar%
	[{\scriptsize\begin{minipage}[t]{\marginparwidth}
			\raggedleft#1%
		\end{minipage}}]%
		{\scriptsize\begin{minipage}[t]{\marginparwidth}
				\raggedright#1%
			\end{minipage}}%
			\fi}

						\newcommand{\headersmall}[3]{
							\pagestyle{plain}
							\noindent
							\vbox{
								\hbox to 6.28in { \Course \hfill #2} 
								\hrule
								
								\vspace{1mm}
								
								Prof. Patt-Shamir
								
								\vspace{7mm}
								
								\hbox to 6.28in {\Large\bf \hfill #3 \hfill} 
							}
							\vspace*{8mm}
						}
						
						\newcommand{\tauheadersmall}[4]{
							\pagestyle{plain}
							\noindent
							\vbox{
								\vspace{-.3in}
								\hbox to 6.5in {\sf #1 \hfill #3}
								\hrule
								#2\hfill
								\vspace{7mm}
								\hbox to 6.28in {\Large\bf \hfill #4 \hfill}
							}
						}
						
						\newcommand{\lectureheader}[3]{
							\pagestyle{plain}
							\noindent
							\vbox{
								\hbox to 6.28in { \Course
									\hfill #2 } 
								\hrule
								\vspace*{1mm}
								Lecturer: #3
								\vspace{7mm}\\
								\hbox to 6.28in {\LARGE\bf \hfill Lecture #1 \hfill} 
							}
							\vspace*{4mm}
						}
						
						\newcommand{\newlectureheader}[3]{
							\pagestyle{plain}
							\noindent
							\vbox{
								\hbox to 6.28in { 6.04s Design and Analysis of Distributed Protocols
									\hfill #2} 
								\hrule
								\vspace*{1mm}
								Lecturer: #3
								\vspace{7mm}\\
								\hbox to 6.28in {\LARGE\bf \hfill Lecture #1 \hfill} 
							}
							\vspace*{4mm}
						}

						\newcommand{\talkheader}[3]{
							\pagestyle{plain}
							\vbox{
								\hbox to 6.28in {Lecturer: #3 \hfill #2 } 
								\hrule
								\vspace*{7mm}
								\hbox to 6.28in {\LARGE\bf \hfill #1 \hfill} 
							}
							\vspace*{4mm}
						}

						\def\Plus{\hbox{\raise 0.3ex\hbox{\tiny +}}} 
						
						\def\mtoday{\ifcase\month\or
							January\or February\or March\or April\or May\or June\or July\or
							August\or September\or October\or November\or
							December\fi\space\number\year}

\usepackage{times}

\usepackage{wrapfig}
\usepackage{xspace}
\usepackage{color}
\usepackage{soul}
\usepackage{hyperref}
\usepackage{amsmath}
\usepackage{amsfonts}
\usepackage{graphicx}
\usepackage{paralist}
\hypersetup{colorlinks,linkcolor=blue,filecolor=blue,citecolor=blue,urlcolor=blue,pdfstartview=FitH}

\newcommand{\namedref}[2]{\hyperref[#2]{#1~\ref*{#2}}}
\newcommand{\sectionref}[1]{\namedref{Section}{#1}}

\newcommand{\theoremref}[1]{\namedref{Theorem}{#1}}

\newcommand{\figref}[1]{\namedref{Figure}{#1}}
\newcommand{\lemmaref}[1]{\namedref{Lemma}{#1}}

\newcommand{\corollaryref}[1]{\namedref{Corollary}{#1}}

\newcommand{\equalityref}[1]{\hyperref[#1]{Equality~(\ref*{#1})}}
\newcommand{\inequalityref}[1]{\hyperref[#1]{Inequality~(\ref*{#1})}}

\newcommand{\claimref}[1]{\namedref{Claim}{#1}}

\newcommand{\CONGEST}{\textbf{CONGEST}\xspace}

\newcommand{\set}[1]{\left\{#1\right\}}

\newcommand{\xc}{\textsc{\sc mcast}}

\newcommand{\mf}{{\sc \textsc{mf}}}
\newcommand{\ep}{{\sc \textsc{e}\psi}}
\newcommand{\ea}{{\sc \textsc{ea}}}
\newcommand{\mfk}{{\sc \textsc{$k$-mf}}}
\newcommand{\hhh}{\text{h}}
\newcommand{\Exc}{\text{X}}
\newcommand{\EST}{\textsc{X}}
\def\MST{\textrm{MST}}

\def\EA{\ea}
\def\Ep{\ep}
\def\EQ{\mathrm{EQ}}
\def\ID{\mathrm{ID}}
\def\COL{\mathrm{Col}}
\def\Dist{\mathrm{dist}}

\def\Fm{\cF_{m}}
\newcommand{\mv}{{\sc \textsc{mv}}}

\def\True{\textsc{true}\xspace}
\def\False{\textsc{false}\xspace}

\def\bV{\text{\bf v}\xspace}
\def\bP{\text{\bf p}\xspace}

\renewcommand{\paragraph}[1]{\par\medskip\noindent\textbf{#1}}

\begin{document}
\mainmatter              
\title{Proof-Labeling Schemes: Broadcast, Unicast and In Between}
%
%
\author{Boaz Patt-Shamir \and Mor Perry}
\authorrunning{B. Patt-Shamir and M. Perry} 
%
%
\institute{School of Electrical Engineering, Tel Aviv University, Tel Aviv 
	6997801, Israel.}

\maketitle              
\begin{abstract}
	We study the effect of limiting the number of different messages a
	node can transmit simultaneously on the verification complexity of
	proof-labeling schemes (PLS). 
	In a PLS, each node is given a label, and  the goal is to verify, by exchanging messages over each link
	in each direction, that a certain global predicate 
	is satisfied by the system configuration.
	We consider a single
	parameter $r$ that bounds the number of distinct 
	messages that can be sent concurrently by any node: in the case
	$r=1$, each
	node may only send the same message to all its neighbors
	(the broadcast model), in the case $r\ge\Delta$, where $\Delta$ is
	the largest node degree in the system, each neighbor may be sent a
	distinct message (the unicast model), and in general, for $1\le
	r\le\Delta$, each of the $r$ messages is
	destined to a subset of the neighbors.
	
	We show that message compression
	linear in $r$ is possible for verifying fundamental
	problems such as the agreement between edge endpoints on the edge state.
	Some problems, including verification of maximal matching, exhibit
	a large gap in complexity between $r=1$ and $r>1$.
	For some other important predicates, the
	verification complexity is insensitive to $r$, e.g., the question whether a
	subset of  edges constitutes a spanning-tree. 
	We also consider the congested clique model.
	We show  that the crossing technique \cite{BFP} for proving lower bounds on the
	verification complexity can be applied in the case of congested
	clique only if $r=1$. Together with a new upper bound, this allows us to
	determine the verification complexity of MST in the broadcast clique.
	Finally, we establish a general connection between the deterministic and randomized 
	verification complexity for any given number $r$.
	\keywords{verification complexity, proof-labeling schemes, CONGEST model, congested clique}
\end{abstract}


\section{Introduction}

Similarly to classical complexity theory, studying the verification
complexity of various problems is one of the major approaches in 
the quest to understand the complexity of network tasks.
The basic idea, proposed by Korman, Kutten and Peleg~\cite{KKP} under
the name Proof-Labeling Schemes (PLS for short), is to
assume that an oracle assigns a label to each node, so that by
exchanging these labels, the nodes can collectively verify that a
certain global predicate holds (see Sec.~\ref{sec-prel} for details). 
The 
verification complexity of a
predicate $\pi$ is defined to be the minimal label length which 
suffices to
verify $\pi$.
This node-centric, space-based view was generalized
in subsequent work, in
which it was allowed for nodes to send different messages to different
neighbors, rather than the whole local label to all
neighbors. Specifically, in \cite{BFP} the verification complexity is
defined to be the minimal \emph{message}-length required to verify the
given predicate. 

The distinction between these two models is natural
and appears in other contexts as well, like the broadcast and the
unicast flavors of congested clique, proposed by Drucker et
al.~\cite{DruckerKO-14}: in the unicast flavor, a node may send a
different message to each of its neighbors, while in the broadcast
flavor, all neighbors receive the same message. Following up on this
model, Becker et 
al.~\cite{BeckerARR16} proposed considering a spectrum of congested
clique models,  where a node may send up to $r$ distinct messages in a
round, where $1\le r<n$ is a given parameter. This model,
called henceforth
$\xc(r)$,  can be motivated by observing that $r$ 
can be viewed as the number
of network interfaces (NICs) a node possesses: Each interface may be
connected to a subset of the neighbors, and it can send only a single
message at a time.

\paragraph{Our Results.} 
In this paper we present a few preliminary results concerning PLS in 
the $\xc(r)$ model. Our main focus is 
on the tradeoff between the number~$r$ of different messages a node can send in one round and 
the verification complexity (message length)~$\kappa$.
While there are problems whose verification complexity is independent of $r$,
we prove that the verification complexity of  some
fundamental problems is highly dependent on $r$.
First, we consider the problem of \emph{matching verification}~(\mv), 
where every node has at most one incident edge
marked, and the goal is to verify whether the set of marks
implies a well defined matching, i.e., an edge is either marked
in both endpoints or unmarked in both, and that this set is a matching.
In~\cite{GS11}, among other results, it is shown that maximal matching 
has verification complexity $\Theta(1)$, and that the verification
complexity of maximum matching in bipartite graphs is also $\Theta(1)$.
These results implicitly assume that the subset of edges is
well defined; our results show that in fact, the main difficulty
is in ensuring that both endpoints of an edge agree on its status. 
This motivates our next problem that focuses on consistency.
Specifically, we define the primitive problem \emph{edge 
	agreement}~(\EA) as follows.
Each node has a $b$-bit string for each incident edge, and a state is 
considered legal 
iff both
endpoints of each edge agree on the string associated with that edge. 
It turns out that the \emph{arboricity} of the graph, denoted 
$\alpha(G)$,
plays an important role in the verification complexity of $\EA$ (and
all problems that $\EA$ can  locally be reduced to).
In \theoremref{EA tight trade-off}, we prove that $\kappa(\EA)\cdot r\in\Theta(\alpha(G)b)$. 
Next, as a more sophisticated example, we
consider the important problem
of \emph{maximum flow}~(\mf):
In \theoremref{flow verification} we show that  $\kappa(\mf)\cdot 
r\in\Theta(\alpha(G)\log f_{\max})$, where $f_{\max}$
is the largest flow value over an edge. 
In~\cite{KKP}, 
a scheme to verify that the maximum flow between a given pair of nodes 
$s$ and $t$ is exactly $k$ is given in the broadcast model, with complexity
$O(k(\log k+\log n))$. We prove, in \theoremref{mfk upper}, that the verification complexity of this problem in the broadcast model is $O(\min\set{\alpha(G),k}(\log k+\log \Delta))$, which is an exponential improvement in some cases. In addition, our upper bound scales linearly with $r$ in the $\xc(r)$ model.  

We also consider the congested 
clique model. 
To date, no lower bounds on the 
verification complexity in the congested clique were known. We show that 
the known technique of 
crossing \cite{BFP} can be applied, but only in 
broadcast clique (i.e.,  $\xc(1)$).
We use this argument, along with a new scheme, to obtain a tight ${\Theta(\log n + \log 
	w_{\max})}$  bound for  MST verification in broadcast cliques, 
where $w_{\max}$ denotes the largest edge weight.

Finally, we show that all results translate to randomized PLS (RPLS)~\cite{BFP}.
Extending a result of \cite{BFP}, we 
show that if both PLS and RPLS are using the same number $r$, then 
an 
exponential difference in verification complexity holds in both 
directions, i.e., in the $\xc(r)$ model, an RPLS with 
verification complexity $O(\log{\kappa_d})$ can be constructed out of 
every PLS  with verification complexity $\kappa_d$, and every RPLS  
with verification 
complexity $\kappa_r$ can be used to construct a PLS with 
verification complexity $O(2^{\kappa_r})$.

\paragraph{Related Work.}
%
Drucker et al.~\cite{DruckerKO-14} propose a \emph{local broadcast} 
communication in the congested clique,
where every node broadcasts a message to all other nodes in each round.
Becker et al.~\cite{BeckerARR16} proposed, still for congested cliques,
a bounded number $r$ of different messages a node can
send in each round.

Verification of a given property in decentralized systems finds applications in various
domains, such as, checking the result obtained from the execution of
a distributed program~\cite{APV91,FRT13}, establishing lower bounds
on the time required for distributed approximation~\cite{DH+12},
estimating the complexity of logic required for distributed run-time
verification~\cite{FRT14}, general distributed complexity
theory~\cite{FKP13}, and self stabilizing algorithms~\cite{BFP14,KKM}.

The notion of distributed verification in a single round was 
introduced by Korman, Kutten, and Peleg in~\cite{KKP}. The verification complexity of minimum spanning-trees (MST)  
was studied  in~\cite{KK07}.
Constant-round schemes were studied in~\cite{GS11}.   
Verification processes in which the global result is not restricted to be the logical 
conjunction of local outputs had been studied in~\cite{AFIM,AFP13}. The role of unique 
node identifiers in local decision and verification was extensively 
studied in~\cite{FHK12,FGKS13,FHS15}.
Proof-labeling schemes in directed networks were studied in~\cite{FLSW16}, where both one-way and two-way 
communication over directed edges is considered.
Verification schemes for dynamic networks, where edges may appear or disappear after label assignment and 
before verification, are studied in~\cite{FRSW17}.
Recently, a hierarchy of local decision as an interaction between a 
prover and a disprover 
was presented in~\cite{FFH16}.

\paragraph{Paper Organization.}
The remainder of this paper  is organized as follows. 
In \sectionref{sec-prel} we formalize the model and recall some 
graph-theoretic concepts. In \sectionref{sec-tools}  we present two
general techniques that apply to the $\xc(r)$ model. 
In \sectionref{sec-edges} we present results for verification
of matching, edge agreement, and max-flow.
In \sectionref{sec-clique} we present our results for congested cliques.
In \sectionref{sec-RPLS} we analyze the relation between deterministic and randomized PLSs.
We conclude in \sectionref{sec-conc} with some open questions and directions for future work.

\section{Model and Preliminaries}\label{sec-prel}
%
\paragraph{Computational Framework and the $\xc$ Model.}
Our model is derived from the \CONGEST model~\cite{peleg:book}. 
Briefly, a distributed network is modeled as a connected undirected
graph $G=(V,E)$, where $V$ is the set of nodes, $E$ is the set of edges, and every node has a unique identifier. In each synchronous round every node performs a local computation, sends a message to each of its neighbors, and receives messages from all neighbors. We denote the number of nodes $|V|$ by $n$ and the 
number of edges $|E|$ by $m$.
For every node $v\in V$, let $d(v)$ be the \emph{degree} of $v$. We denote by $\Delta(G)$ the maximal degree of a node in $G$.
We assume 
that the edges incident to a node $v$ are numbered  $1,\dots,d(v)$.

The main difference between the model considered in this paper,
called  $\xc(r)$, and \CONGEST, is that in $\xc(r)$ we are given a 
parameter 
$r\in\mathbb{N}$ such that a node may send at most $r$ distinct 
messages simultaneously. More precisely, 
we assume that prior to sending messages,
the neighbors of a node are partitioned into $r$ disjoint
subsets (some of which
may be empty), such that $v$ sends the same message to all neighbors in 
a subset.
We emphasize that in our model, for simplicity, $r$ is a uniform parameter for all nodes.

%
%

\paragraph{Proof-Labeling Schemes in the $\xc$ model.}
A \emph{configuration} $G_s$ includes an underlying graph $G=(V,E)$ and a \emph{state} assignment function $s:V\to S$, where $S$ is a (possibly infinite) state space. The state of a node $v$, denoted $s(v)$, includes all local input to $v$. In particular, the state usually includes a unique node identity $\ID(v)$ and, in the case of weighted graphs, the weight $w(e)$ of each
incident edge $e$. The state of $v$ typically
include additional data whose integrity we would like to verify.
{For example, node state may contain a marking of incident edges, 
	such that the set of marked edges constitutes a spanning tree.}

Let $\cF$ be a family  of 
configurations, and let $\cP$ be a boolean predicate  over $\cF$.  A 
proof-labeling
scheme consists of two conceptual components: a \emph{prover} $\bP$, 
and a
\emph{verifier} $\bV$. The prover is an oracle which, given any
configuration $G_s\in \cF$ satisfying $\cP$, assigns a bit string $\ell(v)$ to every
node $v$, called the \emph{label} of $v$. The verifier is a
distributed algorithm running at every node. At each
node $v$, the local verifier takes as input the state $s(v)$ of $v$, its 
label
$\ell(v)$ and based on them sends messages to all neighbors. Then, 
using as input 
the messages received from the neighbors, the local state and the local label,
the local verifier computes a
boolean value.  If the outputs are $\True$ at all nodes, 
the global verifier $\bV$ is said to \emph{accept} the configuration, 
and otherwise (i.e., at least one local verifier outputs $\False$),
$\bV$~is said to \emph{reject} the configuration. For correctness, a 
proof-labeling
scheme $\Sigma=(\bP,\bV)$ for $(\cF,\cP)$ must satisfy the following
requirements, for every $G_s\in\cF$:
\begin{compactitem}
	\item If $\cP(G_s)=\True$ then, using the labels assigned by $\bP$, the verifier $\bV$ accepts $G_s$.
	\item If $\cP(G_s)=\False$ then, for every label assignment, the verifier $\bV$ rejects $G_s$.
\end{compactitem}
Given a configuration $G_s$, we denote by $\vec{c}_\Sigma(G_s)$ the vector of length $|E|$ that contains the messages sent according to the scheme $\Sigma$, and we refer to this vector as the \emph{communication pattern} of $\Sigma$ over $G_s$.
For an underlying graph $G$, we denote by $L(G)$ the 
number of legal configurations of $G$, and by $W_\Sigma(G)$  the 
number of different communication patterns of $\Sigma$ in $G$, over 
all legal configurations. 
In our analysis, given an  edge $(v,u)\in E$, we denote by 
$M_v(e)$ the message over $e$ from $v$ to $u$.

Our central measure for PLSs is its verification complexity, defined as 
follows.
\begin{definition}
	The \emph{verification complexity} of a proof labeling scheme 
	$\Sigma=(\bP,\bV)$ for the predicate $\cP$ over a family of 
	configurations $\cF$ is the maximal length of a message generated 
	by the verifier~$\bV$ based on the labels assigned to the nodes by 
	the prover~$\bP$ in  
	a configuration $G_s$ for which $\cP(G_s)=\True$. 
\end{definition}


In this paper we consider PLSs in the {$\xc(r)$ model}, 
namely we impose the additional restriction that at 
most $r$ distinct messages
may be sent by a node.

\paragraph{Arboricity, degeneracy and average degree.}
The average degree of a graph plays a central role in our study.
However, graphs may have dense and sparse regions. We therefore
use the following refined concepts.
\begin{definition}\label{def-arb}
	The \emph{arboricity} of a graph $G=(V,E)$, denoted by $\alpha(G)$, is defined as the minimum number of acyclic subsets of edges that cover $E$.
	%
	%
	The \emph{degeneracy} of a graph $G$, denoted by $\delta(G)$, is defined as the smallest value $i$ such that the edges of $G$ can be oriented to form a directed acyclic graph with out-degree at most $i$.
\end{definition}

The following properties are well known \cite{NW61,NW64}.
\begin{lemma}\label{lem-arb-deg}
	For all graphs $G$, $\alpha(G)\le\delta(G)<2\alpha(G)$.
\end{lemma}
\begin{lemma}\label{lem-arb}
	For a given graph $G=(V,E)$, 
	$\alpha(G)=\max\Set{\ceil{m_H\over n_H-1}\mid V_H\subseteq 
		V,|V_H|\ge2}$, where 
	$m_H=|E_H|$ and ${n_H=|V_H|}$ over all induced subgraphs 
	$H=(V_H,E_H)$ of $G$.%
	\footnote{
		Given a graph $G=(V,E)$, the \emph{induced subgraph} $H=(V_H,E_H)$ over 
		the set of nodes $V_H\subseteq V$ 
		satisfies that $E_H=E\cap (V_H\times V_H)$.}
\end{lemma}
Note that by Lemmas \ref{lem-arb-deg} and \ref{lem-arb}, the minimal 
number of outgoing edges in the best orientation of a graph $G$
is proportional
to the maximal average degree over all induced subgraphs of $G$.

\section{Techniques for the $\xc$ Model}\label{sec-tools}
In this work, we consider problems expressible as a conjunction of
edge predicates,  where a node may have a different input for every 
edge. 
We present two techniques that can be used as building blocks in the 
design of efficient PLSs in the $\xc$ model.

The first technique, which we call \emph{minimizing orientation},
reduces the number of incident edges a node sends its input on.
We orient the edges such that the maximum out degree is minimized. 
\lemmaref{lem-arb-deg} ensures that the maximum out degree is 
bounded by $2\alpha$. Using a minimizing orientation, we can prove the 
following lemma.

\begin{lemma}\label{lem-orientation}
	Suppose that a verification problem $(\cF,\cP)$ is expressible
	as a conjunction of edge predicates, each involving variables from 
	a single pair of neighbors. Then there exists a PLS 
	$\Sigma=(\bP,\bV)$ for $(\cF,\cP)$ in the $\xc(2\alpha)$ model with 
	verification complexity $k$, where $k$ is the length of the largest local input 
	to an edge predicate.
\end{lemma}

\Proof
We describe the scheme $\Sigma=(\bP,\bV)$.
The prover $\bP$ 
orients the edges such that the maximum out-degree of a node is minimized
(this can be done in linear time, see, e.g., \cite{MB83}).
Then, every node sends its input only on outgoing edges.

The orientation can be verified simply by a local verification that a node
receives a message from every incoming edge and does not receive a message on every outgoing edge.
Note that the empty message is also counted in the number $r$ of different messages. By definition of degeneracy, the maximum out degree is $\delta$, and by \lemmaref{lem-arb-deg} $\delta$ is strictly smaller than $2\alpha$. Therefore, in the $\xc(2\alpha)$ model, in addition to sending a different message on every outgoing edge, we may use the empty message for incoming edges.
After verifying that the orientation is correct (i.e., consistent between neighbors), by definition of the scheme, the
function of every edge is computed by one of its endpoints and the verification is completed.
\QED


\emph{Color addressing.} In the unicast model, each node receives its
own message. However, if we want to use a unicast PLS  in the $\xc(r)$
model with $r<2\alpha$, we may need to bundle
together a few messages, and hence we need to somehow tag each part
of the message with its intended recipient.
Clearly this can be done by tagging
each sub-message by the unique ID of 
recipient, but this adds $\Theta(\log n)$ bits to each sub-message.
The \emph{color addressing} technique reduces this overhead to 
$O(\log\Delta)$. The idea is that each node need only distinguish 
between its neighbors.\footnote{
	We note that using simple port numbering
	requires agreement with the neighbors, which is costly, as we prove in 
	\theoremref{EA tight trade-off}.
}
We solve this difficulty by coloring the nodes so that no two neighbors 
of a node get the same 
color. Formally,
\emph{color addressing} is a PLS $\Sigma_{COL}=(\bP,\bV)$ in the 
broadcast model, where the prover $\bP$ first colors
the nodes so that  no two 
nodes at distance 1 or 2 receive the same color. This is
possible using at most $\Delta^2+\Delta+1\in O(\Delta^2)$ colors, 
because every node has at most 
$\Delta$ 
neighbors and $\Delta^2$ 
nodes at distance $2$ from it.  
Next, the prover assigns to every incident edge of a node the color of the neighbor
at the other end of the edge. 
The verifier $\bV$ at a node $v$ broadcasts the color assigned to $v$ 
by the prover.
Every node verifies that every incident edge is assigned a different color and that the
color received from every edge is the color assigned by the prover to this edge.

Clearly, $\Sigma_{COL}$ guarantees a proper coloring as desired to use 
for addressing, and this coloring is locally verifiable.
Moreover, since a color can be represented using $O(\log \Delta)$ bits, 
we obtain local addressing with verification complexity $O(\log\Delta)$
in the \emph{broadcast} model.
We summarize in the following lemma.

\begin{lemma}\label{lem-coloring}
	$\Sigma_{COL}$ is a PLS in the broadcast model, which assigns and verifies an $O(\log\Delta)$-bit coloring for proper addressing. The verification complexity of
	$\Sigma_{COL}$ is  $O(\log\Delta)$. 
\end{lemma}

\section{Verification Complexity Trade-offs in the $\xc(r)$ 
	Model}\label{sec-edges}
In this section, we study the effect of $r$ on the verification 
complexity of PLSs in the $\xc(r)$ model. We start with the observation 
that for some problems, the asymptotic verification 
complexity is independent of~$r$.
%
These problems include the deterministic verification of a 
spanning-tree and vertex bi-connectivity, and the randomized 
verification of an MST. For each of these problems, we provide a  
scheme for $r=1$ with verification complexity that matches the lower 
bound for $r=\Delta$ \cite{KKP,BFP}.
%
In contrast, there  are problems for which the verification complexity  
is sensitive to $r$. Specifically,
we present a tight bound for the matching verification problem in the 
broadcast model, which is reduced dramatically even for $r=2$.
Finally, we show tight bounds for the primitive problem of 
edge agreement and the more sophisticated application of maximum flow, which scales linearly with $r$.

\subsection{Verification of Matchings}
\label{sec-matching}

In the literature, in verification problems of the form ``does a 
subset of edges satisfy a specified property,'' it is 
usually assumed  that the subset of edges is well defined, i.e., for 
every edge $e=(u,v)$, the local state of $v$ indicates that $e$ is in 
the subset if and only if the local state of $u$ indicates it.
However, since edges do not have storage, an edge set is actually 
represented by the local state at the nodes, and hence consistency 
between neighbors is not always guaranteed.

In fact, there are problems for which  the verification of consistency 
is the dominant factor of the verification complexity. In particular, 
consider matching problems: maximal matching, and maximum matching in 
bipartite graphs. Both problems are known to have constant verification 
complexity~\cite{GS11}. However, these results make the problematic 
assumption that the
edge set in question is well defined.
We consider the matching verification problem using the following
definition.

\begin{definition}[Matching Verification(\mv)]\hfill\\
	\textbf{Instance}: At each node $v$, at most one edge is marked. 
	We use $I_v(e)\in\Set{\True,\False}$ to denote whether $e$ is 
	marked in $v$. \\
	\textbf{Question}: Is the set $M$ of marked edges well defined, i.e., $I_v(e)=I_u(e)$ for every 
	edge $e=(u,v)\in E$, and $M$ is a matching? 
\end{definition}

We argue that in the broadcast model, the verification complexity of this problem is $\Theta(\log\Delta)$.
Formally, we study the problem $(\Fm,\mv)$, where  $\Fm$ is the family of connected configurations with edge indication at each node. We obtain the following result.

\begin{theorem}\label{MV broadcast}
	The verification 
	complexity of $(\Fm,\mv)$ in the broadcast model is 
	$\Theta(\log \Delta)$.
\end{theorem}

We start with proving the lower bound of the theorem as stated in the following lemma. The proof uses a variant of crossing arguments \cite{BFP}.

\begin{lemma}\label{MV broadcast lower bound}
	The verification 
	complexity of any PLS for $(\Fm,\mv)$ in the broadcast model is 
	$\Omega(\log \Delta)$.
\end{lemma}

\Proof
By contradiction. Let $n$ and $2\le\Delta\le n/2+1$ be given.
We construct the following graph $G=(V,E)$ with maximum degree 
$\Delta$. The $n$-node graph $G$
consists of two parts. 
One part is a complete bipartite graph over two sets 
of nodes $A$ and $B$ of size $\Delta-1$ each:  
$H=(A,B,E_H)=K_{\Delta-1,\Delta-1}$; the second part consists of 
an $n-2(\Delta-1)$ path, connected by an edge to a node in $A$. 

Given a configuration $G_s$ of $G$, let $I^{G_s}_v(e)$ denote $I_v(e)$ 
(the mark of edge $e$ as represented in the state of $v$).
Given a configuration 
$H_{s}$ of
the nodes of $H$, extend it to a configuration $G({H_s})$ of $G$ as
follows. For every $e=(u,v)\in E_H$ let 
$I^{G({H_s})}_v(e)=I^{H_{s}}_v(e)$, and for every $e=(u,v)\in 
E\setminus 
E_H$, let $I^{G({H_s})}_v(e)=0$. 
Clearly, $H_s$ is legal if and only if  
$G({H_s})$ is legal. 
We note that the number of different matchings in $H$, $L(H)$, is at
least $(\Delta-1)!$, because
every permutation of $\Delta-1$ elements represents a different 
matching in $H$. 

Now, let $\Sigma=(\bP,\bV)$ be 
a PLS for $(\Fm,\mv)$ in the broadcast model, and assume
for contradiction that 
$\kappa(\Sigma)<\frac{1}{2}\log (\Delta-1)-1$.
Recall that $W_\Sigma(H)$ is the 
number of different communication patterns of $\Sigma$ in $H$. Then 
%
$$
W_\Sigma(H) 
~\stackrel{(1)}{\le}~ 2^{2(\Delta-1)\kappa}\
~\stackrel{(2)}{<}~ 2^{(\Delta-1)\log\frac{\Delta-1}{e}}
~=~ \left(\frac{\Delta-1}{e}\right)^{\Delta-1}
~\le~ (\Delta-1)! \le L(H)~.
$$
Inequality (1) is true since for every PLS in the broadcast model with 
verification complexity $\kappa$, 
every communication pattern in $H$ can be 
constructed by choosing a $\kappa$-bit message for each of the 
$2(\Delta-1)$ nodes in $H$.
Inequality (2) follows from  our
assumption that $\kappa<\frac{1}{2}\log (\Delta-1)-1$ and the fact that 
$\log e<2$. 
Therefore, the number of communication patterns of $\Sigma$ in $H$ is 
strictly smaller than the number of legal configurations of $H$.
Therefore, there must be two different legal configurations $G({H_s})$ and 
$G({H_s'})$ with the same communication pattern in $H$. Since $G({H_s})$ and 
$G({H_s'})$ differ only over $E_H$ edges, there must
exist an edge ${e^*=(v,u)\in E_H}$  such that 
$I^{G({H_s})}_v(e^*) = I^{G({H_s})}_u(e^*)
\neq I^{G({H_s'})}_v(e^*) = I^{G({H_s'})}_u(e^*)$. Consider the configuration obtained from $G({H_s})$ with 
$I^{G({H_s'})}_w(e)$ replacing $I^{G({H_s})}_w(e)$ for every node $w\in B$ and every
edge $e=(w,w')\in E_H$. Intuitively, in this configuration, the state of all nodes in $V\setminus B$ is as in $G({H_s})$, and the state of nodes in $B$ is as in $G({H_s'})$. Obviously, this configuration is illegal, because 
$I^{G({H_s})}_u(e^*)\neq I^{G({H_s'})}_v(e^*)$, and $e\in A\times B$. However, since all 
nodes in $H$ send the same messages in $G({H_s})$ and in $G({H_s'})$ 
under $\Sigma$, we get the following. With the labels assigned by $\bP$ to the set of nodes $B$ in 
$G({H_s'})$ and the labels assigned by $\bP$ to all $V\setminus 
B$ nodes in $G({H_s})$, since all edges connected to $B$ are in $E_H$, the local view of $B$ in verification 
is exactly as in $G({H_s'})$, and the local view of all other nodes 
in verification is exactly as in $G({H_s})$. Therefore, all nodes 
output \True on an illegal configuration, which contradicts the correctness of $\Sigma$.
\QED

The following lemma shows a matching upper bound for this problem. This completes the proof of \theoremref{MV broadcast}.

\begin{lemma}\label{MV broadcast upper bound}
	There exists a PLS $\Sigma=(\bP,\bV)$ for $(\Fm,\mv)$ in the broadcast model with verification complexity 
	$O(\log \Delta)$.
\end{lemma}

\Proof
The constructed scheme $\Sigma=(\bP,\bV)$ uses color addressing.
Let $\COL_v$ be the color assigned to node $v$. The verifier $\bV$ at every node $v$ locally
verifies that it has at most one marked incident edge. If none of the edges of $v$ is marked, then it
sends the empty message, and if there is a marked edge  $(v,u)$, then the verifier at node $v$ sends
$\COL_u$.
Finally, locally verify consistency at every node $v$ as follows. The message received from edge
$(v,w)$ is $\COL_v$ if and only if edge $(v,w)$ is marked.

We now prove the correctness of $\Sigma$.
According to the correctness of
color addressing, every node reliably broadcasts the indication of its marked edge if any.
If the marking is consistent and indicates a matching, then on every marked edge, every endpoint  sends the color of the other endpoints, and all nodes output \True.
If there exists a node with more than one marked edge, by definition of the scheme, this node outputs \False.
Finally, if every node has at most one marked edge but the marking is inconsistent, then there exists an edge $e=(v',u')$ such that $M_{v'}(e)=0$ and $M_{u'}(e)=1$. By definition of the scheme, $u'$ broadcasts $COL_{v'}$, and $v'$ receives its color from an unmarked edge. Therefore, $v'$ outputs \False. 
\QED

The result above says that in the broadcast model, the verification 
complexity of the maximal matching problem and the maximum matching in 
bipartite graphs is dominated by the consistency verification. 
Observe that in the $\xc(2)$ model, the verification 
complexity  of $(\Fm,\mv)$  is $O(1)$, by letting
every node $v$  send on 
every edge $e=(v,u)$ the bit $I_v(e)$: only two types of messages are 
needed! 

We also note that for the problem of maximum matching in cycles, the 
asymptotic verification complexity is unchanged if we must verify 
consistency, since the verification complexity of this problem in the 
broadcast model is $\Theta(\log n)$~\cite{GS11}.

\subsection{The Edge Agreement Problem}
Motivated by the results for matching verification, we now
formalize and study the fundamental problem of consistency across edges.
\begin{definition}[$b$-bit Edge Agreement ($\EA_b$)]\hfill\\
	\textbf{Instance}: Each node 
	$v$ holds in its state a $b$-bit string $B_v(e)$ for each incident 
	edge $e$. \\
	\textbf{Question}: Is $B_v(e)=B_u(e)$ for every 
	edge $e=(u,v)\in E$?
\end{definition}

Let $\cF$ be the family of all configurations, and let $\alpha$ denote 
the arboricity of the graph.
Our first main result is the following tight trade-off between  
$r$ (the number of different messages for a node) 
and verification complexity of $\EA_b$.

\begin{theorem}\label{EA tight trade-off}
	Let $b\in \Omega(\log \Delta)$.
	For every $1\le r \le \min\set{\Delta,2^{b/4}}$, the verification 
	complexity of $(\cF,\EA_b)$ in the $\xc(r)$ model is 
	$\Theta(\ceil{\frac{\alpha}{r}}b)$.
\end{theorem}

This theorem states both an upper and a lower bound. We start with
the lower bound.

\begin{lemma}\label{EA trade-off lower}
	For every $1\le r \le \min\set{\Delta,2^{b/4}}$, the verification 
	complexity of any PLS for $(\cF,\EA_b)$ in the $\xc(r)$ model is 
	$\Omega((\frac{\alpha}{r}+1)b)$.
\end{lemma}
To prove \lemmaref{EA trade-off lower}, we prove the following claim 
using ideas similar to those used in the proof of 
\lemmaref{MV broadcast lower bound}.

\begin{claim}\label{conf crossing}
	Let $G=(V,E)$ be a graph, let $1\le r \le 
	\min\set{\Delta,2^{b/4}}$  and consider a PLS for $(\cF,\EA_b)$ in the $\xc(r)$ model. 
	For every induced subgraph $H=(V_H,E_H)$ of $G$,  $W_\Sigma(H) \ge L(H)$.
\end{claim}

\Proof
Given $1\le r \le \min\set{\Delta,2^{b/4}}$, let $\Sigma=(\bP,\bV)$ be 
a PLS for $(\cF,\EA_b)$ in the $\xc(r)$ model.  Given $G=(V,E)$, 
let $H=(V_H,E_H)$ be an induced subgraph of $G$.
Let $B^{G_s}_v(e)$ be $B_v(e)$ (the bit-string held by $v$ for 
the edge $e$) in configuration $G_{s}$. Given a configuration 
$H_{s}$ of
the nodes of $H$, extend it to a configuration $G({H_s})$ of $G$ as
follows: for every $e=(u,v)\in E_H$ let 
$B^{G({H_s})}_v(e)=B^{H_{s}}_v(e)$, and for every $e=(u,v)\in 
E\setminus 
E_H$, let $B^{G({H_s})}_v(e)=0^b$. 
Clearly, $H_s$ is legal if and only if  
$G({H_s})$ is legal. 

Now, to prove the claim, assume for contradiction that
$W_\Sigma(H)<L(H)$.  
Then there must be two different legal configurations $G({H_s})$ and 
$G({H_s'})$ with the same communication pattern.  There must
exist an edge ${e^*=(v,u)\in E_H}$  such that 
$B^{G({H_s})}_v(e^*) = B^{G({H_s})}_u(e^*)
\neq B^{G({H_s'})}_v(e^*) = B^{G({H_s'})}_u(e^*)$:
This is because we assume $G({H_s})\ne G({H_s'})$, and by 
construction, the difference can be only in $E_H$ 
edges. Consider the configuration obtained from $G({H_s})$ with 
$B^{G({H_s'})}_v(e)$ replacing $B^{G({H_s})}_v(e)$ for every 
edge $e=(v,w)\in E$. This configuration is illegal, because 
$B^{G({H_s})}_u(e^*)\neq B^{G({H_s'})}_v(e^*)$. However, since all 
nodes send the same messages in $G({H_s})$ and in $G({H_s'})$ 
under $\Sigma$, we get that with the label assigned by $\bP$ to $v$ in 
$G({H_s'})$ and the labels assigned by $\bP$ to all $V\setminus 
\set{v}$ nodes in $G({H_s})$, the local view of $v$ in verification 
is exactly as in $G({H_s'})$, and the local view of all other nodes 
in verification is exactly as in $G({H_s})$. Therefore, all nodes 
output \True on an illegal configuration, a contradiction.
\QED

\ProofOf{Lemma}{EA trade-off lower}
It is known   that
the non-deterministic two-party communication complexity of verifying
the equality ($\EQ$) of $b$-bit strings is $\Omega(b)$ \cite[Example 2.5]{KN}. 
Simulating a verification scheme for $(\cF,\EA_b)$ on a network of one 
edge, is a correct non-deterministic
two-party communication protocol for $\EQ$. 
Therefore, $\Omega(b)$ is a lower bound for $(\cF,\EA_b)$. 

We now prove that $\Omega(\frac{\alpha}{r}b)$ is also a lower bound for $(\cF,\EA_b)$. 
Let $G_s\in \cF$ be a configuration with an underlying graph $G=(V,E)$, 
and let $H=(V_H,E_H)$ be the densest induced subgraph of $G$, i.e., 
${m_H/n_H\ge m_{H'}/n_{H'}}$ for every $V_H'\subseteq V$. By 
\lemmaref{lem-arb}, $\alpha=\lceil m_H/(n_H-1)\rceil$. W.l.o.g., let
$V_H=\{v_1,\ldots,v_{n_H}\}$, and let ${d_H(v_i)=|\{(v_i,v_j)\in E_H\}|}$ 
be the degree of node $v_i$ in $H$. 

We now show that for $1\le r \le \min\set{\Delta,2^{b/4}}$ and any 
scheme $\Sigma$ for $(\cF,\EA_b)$ with verification complexity 
$\kappa<\frac{\alpha b}{4r}-2$ in the $\xc(r)$ model, it holds that $W_\Sigma(H)<L(H)$.
Let $\Sigma$ be such a verification scheme. Then
%
%
%
\begin{align}
	\label{eq0}
	W_\Sigma(H) &\le \prod_{i=1}^{n_H} \left[\dbinom{2^\kappa}{r}\cdot 
	r^{d_H(v_i)}\right] \\
	\label{eq1}
	&\le \left(\frac{2^\kappa\cdot e}{r}\right)^{rn_H}\cdot r^{2m_H} \\
	\label{eq2}
	&< 2^{\alpha b n_H/4}\cdot r^{2m_H} \\
	\label{eq3}
	&\le 2^{\frac{b}{2}m_H}\cdot r^{2m_H} \\
	\label{eq4}
	&\le 2^{bm_H} = L(H)~.
\end{align}

\inequalityref{eq0} is true since for every PLS in the $\xc(r)$ model with 
verification complexity $\kappa$, 
every communication pattern can be 
constructed by letting each node $v_i$
choose $r$ different messages of size $\kappa$ each,
and for each of its $d_H(v_i)$ neighbors, let it choose one of the $r$ 
messages to send. \inequalityref{eq1} is due to the fact that 
${x\choose y}\le (\frac{x\cdot e}{y})^y$  for 
$x,y\ge 0$. \inequalityref{eq2} follows from  our
assumption that $\kappa<\frac{\alpha b}{4r}-2$. \inequalityref{eq3} 
follows from \lemmaref{lem-arb} which implies that $\alpha\le 
2m_H/n_H$, and  \inequalityref{eq4}  from our assumption that 
$r \le 2^{b/4}$.

Therefore we may conclude that if $\kappa<\frac{\alpha b}{4r}-2$, then, 
by \claimref{conf crossing}, $\Sigma$ is not a correct 
verification scheme for $(\cF,\EA_b)$. 
This concludes the proof of the lower bound.
%
%
\Qed

Next, we turn to the upper bound. To this end we define a more general 
problem  as follows.

\begin{definition}[$b$-bit Edge $\psi$ ($\Ep_b$)]\hfill\\
	\textbf{Instance}: Each node 
	$v$ holds in its state a $b$-bit string $B_v(e)$ for each incident edge 
	$e$. \\
	\textbf{Question}: Is $\psi_b(B_v(e),B_u(e))=\True$ for every 
	edge $e=(u,v)$, where  $\psi_b$ is a given symmetric predicate of two 
	$b$-bit strings,
	i.e., $\psi_b:\Set{0,1}^b\times\Set{0,1}^b\to\Set{\True,\False}$ and 
	$\psi(s,s')=\psi(s',s)$ for all $s,s'\in\Set{0,1}^b$? 
\end{definition}
\begin{lemma}\label{EA trade-off upper}
	For every $1\le r < 2\alpha$, there exists a PLS for 
	$(\cF,\Ep_b)$ in the $\xc(r)$ model  with verification complexity $O(\frac\alpha{r}(b+\log \Delta))$, and for every $2\alpha\le r \le \Delta$, there exists a PLS for 
	$(\cF,\Ep_b)$ in the $\xc(r)$ model  with verification complexity $O(b)$.
\end{lemma}
\Proof
Let $1\le r < 2\alpha$ be given. We construct a scheme $\Sigma=(\bP,\bV)$ in the $\xc(r)$ model as follows. $\Sigma$ uses color addressing and minimizing orientation. 
Let $\COL_v$ be the color assigned to node $v$, 
and let $d^o(v)$ denote 
the out-degree of node $v$ under the orientation. 
The verifier $\bV$ at node $v$ partitions
the outgoing edges into $r$ parts, each of size at most $Q\DEF\ceil{d^o(v)/r}$, such that 
all edges in a part are sent the same message as follows.
Let $\set{e_1=(v,w_1),\ldots,e_Q=(v,w_Q)}$ be one part of outgoing edges of $v$. 
The message $v$ sends over these edges is the 
list of $Q$ 
pairs $(\COL_{w_i}, B_v(e_i))$ for $1\le i\le Q$. One of the messages that are sent over outgoing edges is sent over all incoming edges (in order to meet the limit of only $r$ different messages.)
Every node $v$, upon 
receiving a message 
$M_w(e)$  over an edge $e=(v,w)$, verifies the following conditions.	
\begin{compactenum}
	\item If $e$ is an outgoing edge, then there exists no pair $(\COL, x)$ in $M_w(e)$ such that $\COL=\COL_v$.
	\item If $e$ is an incoming edge then: 
	\begin{compactenum}
		\item There exists exactly one pair $(\COL, x)$ in $M_w(e)$ such that $\COL=\COL_v$.
		\item For the pair $(\COL_v, x)$, it holds that 
		$\psi_b(x, B_v(e))=\True$.
	\end{compactenum}
\end{compactenum} 

Obviously, this is a PLS in the $\xc(r)$ model.
We now prove its correctness.
By \lemmaref{lem-coloring}, we can assume that the colors of neighbors of each node are different from each other and from the color of the node.
If the configuration is legal and labels are assigned 
according to $\bP$, all nodes output \True. Suppose now that the 
configuration is illegal.
Hence, there must be two neighbors, $v$ and $u$, such that for the edge $e=(u,v)$ we have ${\psi_b(B_u(e),B_v(e))=\False}$. 
Since $r<2\alpha$, \lemmaref{lem-orientation} does not imply a proper verification of the orientation. Therefore, our scheme should verify it. If both $v$ and $u$ consider $e$ as an outgoing edge, then $M_v(e)$ contains the pair $(\COL_u, B_v(e))$. Therefore, $u$ rejects condition (2) and outputs \False.
If both $v$ and $u$ consider $e$ as an incoming edge, then $M_v(e)$ does not contain a pair $(\COL, x)$ such that $\COL=\COL_u$. Therefore, $u$ rejects condition (3.a) and outputs \False.
Assume now, w.l.o.g., that $e$ is oriented from $v$ to $u$.
By definition of the scheme, there exists exactly one pair $(\COL_u, x)$ in $M_v(e)$, and for this pair we know that $x= B_v(e)$.
Since $\psi_b(B_u(e),B_v(e))=\False$, $u$ rejects condition (3.b) and outputs \False.

Regarding complexity, by definition of degeneracy, for 
every $v$, it holds that $d^o(v)\le\delta(G)$. By \lemmaref{lem-arb-deg}, 
$\delta(G)< 2\alpha(G)$.
By \lemmaref{lem-coloring}, every 
$\COL$ can be represented using $O(\log \Delta)$ bits, and overall, every 
message is of size $O(\frac{\alpha}{r}(b+\log \Delta))$. 

For $2\alpha\le r \le \Delta$, by \lemmaref{lem-orientation} and the fact that every scheme in the $\xc(r_1)$ model is in particular a scheme in the $\xc(r_2)$ model for $r_2\ge r_1$, there exists a PLS $\Sigma'=(\bP',\bV')$ for $(\cF,\Ep_b)$ in the $\xc(r)$ model with verification complexity $b$.  
\QED

$\EA_b$ is a special case of $\Ep_b$, where $\psi$ is the equality 
predicate. Therefore, \lemmaref{EA trade-off upper} gives a tight upper 
bound for $(\cF,\EA_b)$ for the case $b\in\Omega(\log \Delta)$. This 
concludes the proof of \theoremref{EA 
	tight trade-off}.

We note that \theoremref{EA	tight trade-off}, in conjunction with 
\theoremref{thm-rpls},
gives  the following corollary.

\begin{corollary}\label{Rand EA tight trade-off}
	Let $b\in \Omega(\log \Delta)$.
	For every $1\le r \le \min\set{\Delta,2^{b/4}}$, the randomized verification 
	complexity of $(\cF,\EA_b)$  in the $\xc(r)$ model is 
	$\Theta(\log (\ceil{\frac{\alpha}{r}}b))$.
\end{corollary}

\subsection{An Advanced Example: The Maximum Flow Problem}
In this section we consider a more sophisticated problem, namely Maximum
Flow in the context of the $\xc(r)$ model.
The best previously known result~\cite{KKP}  was for verification 
of ``$k$-flow'': the goal is to verify that the maximum flow between a 
given pair of nodes is exactly $k$. The verification complexity of the 
scheme of \cite{KKP} is $O(k(\log k+\log n))$ in the broadcast model.
In~\theoremref{mfk upper}, we show an improvement of this result and a generalization to the $\xc(r)$ model.

First, we solve a slightly different problem,
formalized as follows.
Let $\cF_{st}$ be the family of configurations of graphs,
where a graph in  $\cF_{st}$ has two distinct nodes denoted 
$s$ 
and $t$ called \emph{source} and \emph{sink}, respectively,
and a natural number $c(e)$ called the 
\emph{capacity} associated with each edge $e$.
The $\mf$ problem is defined over the family of 
configurations $\cF_{st}$ as follows.
\begin{definition}[Maximum Flow (\mf)]\hfill\\
	\textbf{Instance}: A configuration $G_s\in \cF_{st}$, where each 
	node $v$ has an integer $f(v,u)$ for every 
	neighbor $u$.\\
	\textbf{Question}: Interpreting $f(v,u)$ as the amount of flow
	from $v$ to $u$ ($f(v,u)<0$ means flow  from 
	$u$ to $v$), is $f$ a maximum flow from  
	$s$ to $t$? 
\end{definition}
Recall that $f$ is a legal flow iff it satisfies the 
following 
three conditions (see, e.g., \cite{AhujaMO-93}). 
\begin{compactitem}
	\item Anti symmetry: for every $(v,u)\in E$, $f(v,u) = -f(u,v)$.
	\item Capacity compliance: for every $(v,u)\in E$, $|f(v,u)|\le 
	c(v,u)$.
	\item Flow conservation: for every node $v\in V\setminus\Set{s,t}$, 
	$\sum_{u\in V}f(v,u) = 0$.
\end{compactitem}
If all three conditions hold, 
then, by the max-flow min-cut theorem, $f$ is maximum iff 
there is a saturated cut.

We denote by $f_{\max}$ the maximal flow amount 
over all edges of $G$ (note that $f_{\max}$ need not be polynomial in 
$n$). Also, for a bit string $x=x_0x_1\cdots x_k$, let $\bar 
x=\sum_{i=0}^kx_i2^i$. 

\begin{theorem}\label{flow verification}
	Let $\log f_{\max}\in \Omega(\log n)$.
	There exists a constant $c>1$ such that for	every 
	$1\le r \le \min\set{\alpha/c,\sqrt[4]{f_{\max}}}$, the 
	verification 
	complexity of  $(\cF_{st},\mf)$  in the $\xc(r)$ model is 
	$\Theta(\log (f_{\max})\alpha/r)$.
\end{theorem}

Again, we start with the lower bound. 
We note that the counting argument used for $\EA_b$ (\lemmaref{EA trade-off lower})
cannot be applied to this problem.
To prove the lower bound for $\mf$, we show a non-trivial reduction from a problem in $(\cF,\EA_b)$ to a problem in $(\cF_{st},\mf)$.

\begin{lemma}\label{Flow lower}
	Let $\log f_{\max}\in \Omega(\log n)$. There exists a constant $c>1$ such that for	every 
	$1\le r \le \min\set{\alpha/c,\sqrt[4]{f_{\max}}}$, the 
	verification complexity of any PLS for $(\cF_{st},\mf)$  in the $\xc(r)$ model 
	is $\Omega(\log (f_{\max})\alpha/r)$.
\end{lemma}

\Proof
We reduce the problem $(\cF,\EA_{b})$ to the problem $(\cF_{st},\mf)$ 
with $f_{\max}\le m\cdot 2^b$.
Let $\Sigma_f=(\bP_f,\bV_f)$ be a PLS for $(\cF_{st},\mf)$ in the $\xc(r)$ model. We  construct from $\Sigma_f$ a PLS $\Sigma=(\bP,\bV)$ for $(\cF,\EA_{b})$ in the $\xc(r)$ model. Let $G_s\in \cF$ be such that $\EA_b(G_s)=\True$ with an underlying 
graph $G=(V,E)$, whose arboricity is $\alpha$.
Let $T=(V,E_T)$ 
be a breadth-first spanning tree of 
$G$ rooted at some node $w\in V$.
We denote by $p(v)$ the parent of $v$ in $T$.
Intuitively, the prover constructs from a legal $(\cF,\EA_{b})$ instance $G_s$ a legal $(\cF_{st},\mf)$ instance $G_s'$ by letting exactly one node to  simulate nodes $s$ and $t$  which are connected to the rest of the network with edges of capacity $0$. Assigning a capacity to every edge in $E$ results in $G_s'\in \cF_{st}$. Then, the prover defines a cyclic legal flow according to the input strings $B_v(u)$. Think of an arbitrary orientation of the edges, and assume that for the oriented edge $(v,u)$ with $B_v(u)=B_u(v)=B$ we assign the flow $f(v,u)=-\overline{B}$ and $f(u,v)=\overline{B}$. Of course, $f$ does not necessarily satisfy flow conservation. The actual flow we define would be the sum of $f$ with the convergecast on $T$ of excess flow from all nodes. The result is a legal flow of value $0$ in a network where the minimum cut is $0$. Therefore,  $\mf(G_s')=\True$.

Formally, the prover orients all edges arbitrarily.
Let $\hhh_v(u)$ be the variable indicating to $v$ the orientation of edge $(v,u)$. If the prover orient $e=(v,u)$  from $v$ to $u$ then $\hhh_v(u)=-1$ and $\hhh_u(v)=1$. We define $\Exc(v)$ as the excess flow of node $v$ if the flow over every edge $e=(v,u)$ were $\hhh_v(u)\cdot \overline{B_v(u)}$, i.e., $\Exc(v)=\sum_{(v,u)\in E}\hhh_v(u)\cdot \overline{B_v(u)}$.
We denote by $T_v$ the set of nodes in the sub-tree rooted at  $v$, and define $\EST(T_v)$ as the sum of all  excess flow in this sub-tree, $\EST(T_v)=\sum_{u\in T_v}\Exc(u)$.
Let $G_s'\in \cF_{st}$ be a configuration as follows. Add to $G$ nodes $s$ and $t$ connected to $w$ with edges $(s,w),(t,w)$ of capacity $0$. The idea is that $w$ simulates in addition to the verification of $\mf$ of itself, the verification of $\mf$ of $s$ and $t$. For every edge $e\in E$, $c(e)=m\cdot 2^b$. For every node $v\in V$ and neighbor $u\in V$ the flow $f(v,u)$ in configuration  $G_s'$ is defined as follows. If $u=p(v)$ then $f(v,u)=\hhh_v(u)\cdot \overline{B_v(u)}-\EST(T_v)$, if $v=p(u)$ then $f(v,u)=\hhh_v(u)\cdot\overline{B_v(u)}+\EST(T_u)$, and otherwise $f(v,u)=\hhh_v(u)\cdot \overline{B_v(u)}$.
For completeness, we define $f(s,w)=f(w,s)=f(t,w)=f(w,t)=0$.

We first prove that if $\EA_{b}(G_s) = \True$ then $\mf(G_s')=\True$. 
By construction, the value of flow going out of $s$ and into $t$ is $0$, and the capacity of the cut $(\set{s};V\cup \set{t})$ is $0$. Therefore, if $f$ is a legal flow then it is a maximum flow.
Asymmetry holds by construction and the assumption that $B_v(u)=B_u(v)$ for every $u$ and $v$.
Capacity constraints hold by definition of $\EST(T_v)$ and the fact that $\overline{B_v(u)}\le 2^b$. Finally, we prove that flow conservation holds. Let $\cC(v)=\set{u\mid p(u)=v}$ be the set of children of $v$ in $T$, and let $\cY(v)=\set{u\mid (v,u)\in E, u\neq p(v), v\neq p(u)}$ be the set of neighbors of $v$ which are neither the children of $v$ nor $v$'s parent in $T$. 
The flow at every node $v\neq w$ satisfies the following.

\begin{align}
&\sum_{(v,u)\in E}f(v,u)=f(v,p(v))+\sum_{u\in \cC(v)}f(v,u)+\sum_{u\in \cY(v)}f(v,u)\nonumber \\ 
&=\hhh_v(u)\cdot\overline{B_v(u)}-\EST(T_v)+\sum_{u\in \cC(v)}(\hhh_v(u)\cdot\overline{B_v(u)}+\EST(T_u))+\sum_{u\in \cY(v)}(\hhh_v(u)\cdot \overline{B_v(u)}) \label{Eq.flow}\\
&=\sum_{(v,u)\in E}\hhh_v(u)\cdot \overline{B_v(u)}-\EST(T_v)+\sum_{u\in \cC(v)}\EST(T_u) \nonumber \\
&=\Exc(v)-\EST(T_v)+\sum_{u\in \cC(v)}\EST(T_u)=0 \label{Eq.EST}
\end{align}

\equalityref{Eq.flow} is true by construction of $f$, and \equalityref{Eq.EST} is true since by definition, $\EST(T_v)=\Exc(v)+\sum_{u\in \cC(v)}\EST(T_u)$.

For $w$, we know that $f(w,s)=f(w,t)=0$. Therefore, we only need to prove flow conservation over the edges $(w,u)\in E$.
\begin{align}
&\sum_{(w,u)\in E}f(w,u)=\sum_{u\in \cC(w)}(\hhh_w(u)\cdot\overline{B_w(u)}+\EST(T_u)) \label{Eq.BFS}\\
&=\sum_{(w,u)\in E}\hhh_w(u)\cdot \overline{B_w(u)}+\sum_{u\in \cC(w)}\EST(T_u) \nonumber \\
&=\Exc(w)+\sum_{u\in \cC(w)}\EST(T_u)=\sum_{v\in V}\Exc(v) \label{Eq.allExc}\\
&=\sum_{v\in V}\sum_{(v,u)\in E}\hhh_v(u)\cdot \overline{B_v(u)}=0 \label{Eq.cancel}
\end{align}

\equalityref{Eq.BFS} is true since $T$ is a BFS rooted at $w$ and therefore, every neighbor of $w$ is a child of $w$ in $T$, and \equalityref{Eq.allExc} is true since by definition, $\EST(T_w)=\Exc(w)+\sum_{u\in \cC(w)}\EST(T_u)$ is the sum of all  excess flow of nodes in $T$, and $T$ spans $G$. \equalityref{Eq.cancel} is true by assumption that for every $u$ and $v$, $B_v(u)=B_u(v)$ and $\hhh_v(u)\cdot\hhh_u(v)=-1$.
For $s$ and $t$, flow conservation holds immediately by construction.
This concludes the proof that if $\EA_{b}(G_s) = \True$ then $\mf(G_s')=\True$. 

We now describe the details of the scheme $\Sigma=(\bP,\bV)$ for $(\cF,\EA_{b})$.
Given a configuration $G_s\in \cF$ such that $\EA_{b}(G_s) = \True$, the prover $\bP$ constructs the configuration  $G_s'\in \cF_{st}$ and assigns 
for every node $v\in V$ a label which is 
composed of seven parts: 
\begin{compactitem}
	\item $\ell_{1,v}=\ID_w$.
	\item $\ell_{2,v}=\Dist(w,v)$ is the distance between $w$ and $v$ in $T$.
	\item $\ell_{3,v}=\ID_{p(v)}$.
	\item $\ell_{4,v}=\set{\hhh_v(u)\mid (v,u)\in E}$.
	\item $\ell_{5,v}=\set{f(v,u)\mid (v,u)\in E}$.
	\item $\ell_{6,v}=\EST(T_v)$.
	\item $\ell_{7,v}=\ell_f(v)$ is the label assigned by $\bP_f$ to $v$ in $G_s'$ (if $v=w$ then it receives also $\ell_f(s)$ and $\ell_f(t)$).
\end{compactitem}

The verifier $\bV$ at node $v$ operates as follows. The message sent 
from node $v\in V$ to its neighbor $u\in V$ is 
$M_v(u)=(\ID_v,\ell_{1,v},\ell_{2,v},\ell_{3,v},\ell_{6,v},M^f_{v}(u))$ where 
$M^f_{v}(u)$ is the message $v$ sends to $u$ according to 
$\Sigma_f(G'_s)$. Upon receiving messages $M_u(v)$ from all neighbors, 
$v$ outputs the conjunction of the following. 
\begin{compactenum}
	\item If $\ell_{2,v}=0$ then $\ell_{1,v}= \ID_v$.
	\item $\ell_{1,u}=\ell_{1,v}$ for every neighbor $u$.
	\item If $\ell_{2,v}>0$ then there exists a neighbor $u$ such that $\ell_{2,u}=\ell_{2,v}-1$ and $\ell_{3,v}=\ID_u$.
	\item For $u$ such that $\ell_{3,v}=\ID_u$, $f(v,u)=\hhh_v(u)\cdot \overline{B_v(u)}-\ell_{6,v}$.
	\item For $u$ such that $\ell_{3,u}=\ID_v$, $f(v,u)=\hhh_v(u)\cdot\overline{B_v(u)}+\ell_{6,u}$.
	\item For $u$ such that neither $\ell_{3,v}=\ID_u$ nor $\ell_{3,u}=\ID_v$, $f(v,u)=\hhh_v(u)\cdot\overline{B_v(u)}$.
	\item If $\ell_{2,v}> 0$ then $v$ simulates the output of $v$ according to 
	$\bV_f$ with the set of flows $\ell_{5,v}$, label $\ell_{7,v}$ and the received message $M^f_{u}(v)$ from every neighbor $u$.
	\item If $\ell_{2,v}= 0$ then $v$ simulates the output of $v$ according to 
	$\bV_f$ with the set of flows $\ell_{5,v}\cup\set{f(v,s)=0,f(v,t)=0}$, label $\ell_{7,v}$ and the received message $M^f_{u}(v)$ from every neighbor $u$ and messages $M^f_{s}(v)$ and $M^f_{t}(v)$ that are sent from $s$ with label $\ell_f(s)$ and from $t$ with label $\ell_f(t)$ according to $\bV_f$ respectively. In addition, $v$ simulates the output of $s$ and $t$ according to $\bV_f$ with the flows $f(s,v)=0$ and $f(t,v)=0$, labels $\ell_f(s)$ and $\ell_f(t)$,  and received messages $M^f_{v}(s)$ and $M^f_{v}(t)$ respectively. The result of this verification item is the conjunction of outputs of $v,s$ and $t$.
\end{compactenum}

By construction, if $M^f_v(u_1)=M^f_v(u_2)$ then $M_v(u_1)=M_v(u_2)$. Therefore, if $\Sigma_f$ is a PLS in the $\xc(r)$ model then $\Sigma$ is a PLS in the $\xc(r)$ model.
Let $c^*$ be a constant such that every $\ID$ is at most $n^{c^*}$.
If the verification complexity of $\Sigma_f$ is   
$\kappa$, then the verification 
complexity of $\Sigma$ is $\kappa+b+3(c^*+1)\log n$. This is true because $\EST(T_v)\le n^2\cdot 2^b$. If $\Sigma$ is a correct verification 
scheme for $(\cF,\EA_{b})$, by the proof of \lemmaref{EA trade-off lower}, its verification complexity is greater than $\frac{\alpha}{4r}b-2$. If $b\ge\log n$ and $\frac{\alpha}{r}> 12c^*+16$ we get that $\kappa\in \Omega(\frac{\alpha}{r}b)$. Since by construction ${f_{\max}\le m\cdot 2^b}$, it follows that $\kappa\in\Omega(\frac{\alpha}{r}\log(f_{\max}/n))$, and if $\log f_{\max}\in \Omega(\log n)$ we get  that $\kappa\in\Omega(\frac{\alpha}{r}\log f_{\max})$.

We now prove the correctness of $\Sigma$.
If $\EA_b(G_s)=\True$ and labels are assigned according to $\bP$, then clearly, all nodes output \True.
Suppose now that $\EA_b(G_s)=\False$. Then, there is at least one edge 
$(v,u)\in E$ such that $B_v(u)\neq B_u(v)$. We assume that all nodes output \True and show that it leads to a contradiction. If all nodes verify 
properties (1),(2) and (3) then there is exactly one node $w$ which 
simulates $s$ and $t$. Therefore, the simulated configuration is in 
$\cF_{st}$. If all nodes verify 
properties (7) and (8) then the simulated configuration $\hat{G}_s$ satisfies $\mf(\hat{G}_s)=\True$. In particular, $f(v,u)= -f(u,v)$. If $u$ and $v$ verify properties (4),(5) and (6) then the following holds.
If $\ell_{3,u}=\ID_v$ then $f(v,u)=\hhh_v(u)\cdot\overline{B_v(u)}+\ell_{6,u}$ and $f(u,v)=\hhh_u(v)\cdot \overline{B_u(v)}-\ell_{6,u}$.
If $\ell_{3,v}=\ID_u$ then $f(v,u)=\hhh_v(u)\cdot \overline{B_v(u)}-\ell_{6,v}$ and $f(u,v)=\hhh_u(v)\cdot \overline{B_u(v)}+\ell_{6,v}$.
If neither $\ell_{3,v}=\ID_u$ nor $\ell_{3,u}=\ID_v$ then $f(v,u)=\hhh_v(u)\cdot\overline{B_v(u)}$ and $f(u,v)=\hhh_u(v)\cdot\overline{B_u(v)}$.
In all three possible cases, since $\hhh_v(u)$ and $\hhh_u(v)$ have values either $1$ or $-1$, it holds that $f(v,u)= -f(u,v)$ if and only if $B_v(u)= B_u(v)$. Contradicting the assumption that $B_v(u)\neq B_u(v)$. This concluded the proof.
\QED

\begin{lemma}\label{Flow upper}
	For every $1\le r < 2\alpha$, there exists a PLS for $(\cF_{st},\mf)$
	in the $\xc(r)$ model with verification complexity $O(\frac{\alpha}{r}(\log f_{\max}+\log \Delta))$, and for every $2\alpha\le r \le \Delta$, there exists a PLS for 
	$(\cF_{st},\mf)$ in the $\xc(r)$ model with verification complexity $O(\log f_{\max})$.
\end{lemma}
\Proof
Let $\psi$ be the function of two input strings such that $\psi(x,y)=\True$  iff $\overline{x}=-\overline{y}$, and let $\Sigma_{A}=(\bP_A,\bV_A)$ be a PLS  for 
$(\cF,\Ep_b)$ in the $\xc(r)$ model. We describe the details of a PLS
$\Sigma=(\bP,\bV)$ for $(\cF_{st},\mf)$ in the $\xc(r)$ model. Let $G_s\in \cF_{st}$ be a 
configuration with an underlying graph $G=(V,E)$ and $\mf(G_s)=\True$. 
Consider the configuration $G'_s$ defined as follows. The underlying 
graph of $G'_s$ is $G=(V,E)$ and for every node $v\in V$ and every edge 
$e=(v,u)$, $B_v(e)=f(v,u)$. Obviously, $G'_s\in \cF$. In addition, $\Ep_b(G'_s)=\True$ since, in 
particular, the flow in $G_s$ satisfies asymmetry on every edge.
Let $Z\subset V$ be such that $(Z;V\setminus Z)$ is a minimum $s$-$t$ cut in $G_s$ with $s\in Z$ and $t\notin Z$. The label assigned by $\bP$ to every node $v\in V$ is composed of two parts: $\ell_{1,v}=z(v)$ where $z(v)$ is a bit such that $z(v)=1 \iff v\in Z$, and $\ell_{2,v}=\ell_A(v)$ is the label assigned by $\bP_A$ to $v$ in $G'_s$. 

The verifier $\bV$ at node $v$ operates as follows. The message sent from  $v$ over edge $e=(v,u)$ is $M_v(e)=(\ell_{1,v},M^A_{v}(e))$ where $M^A_{v}(e)$ is the message $v$ sends over $e$  in $\Sigma_A(G'_s)$. Upon receiving a message $M_u(e)$ over every edge $e=(u,v)$, node $v$ outputs the conjunction of the following. 
\begin{compactenum}
	\item The output of $\bV_A$ upon receiving a message $M^{A}_u(e)$ from every neighbor $u$ of $v$, where the label of $v$ is $\ell_{2,v}$.
	\item For every edge $(v,u)$, $|f(v,u)|\le c(v,u)$.
	\item If $v\neq s,t$ then $\sum_{u\in V}f(v,u) = 0$.
	\item If $v=s$ then 
	$\ell_{1,v}=1$.
	\item If $v=t$ then 
	$\ell_{1,v}=0$.
	\item For every neighbor $u$ of $v$, if $\ell_{1,u}\neq \ell_{1,v}$ then $|f(v,u)|= c(v,u)$.
\end{compactenum}

If $\Sigma_A$ is a PLS in the $\xc(r)$ model with verification complexity $\kappa$, then $\Sigma$ is a PLS in the $\xc(r)$ model with verification complexity $\kappa+1$. This is true because if $M^A_v(e_1)=M^A_v(e_2)$ then $M_v(e_1)=M_v(e_2)$ and $\ell_{1,v}$ is one bit. $\Sigma_A$ is a verification scheme for $(\cF,\Ep_b)$ where $b=\log (f_{\max})$.  By \lemmaref{EA trade-off upper}, the upper bounds follow.

We now prove the correctness of $\Sigma$. If $\mf(G_s)=\True$ then for every $v$ and $u$, $f(v,u) = -f(u,v)$. From correctness of $\Sigma_{A}$, the result of $v$ in (1) is \True. The result of $v$ in (2) and (3) is also \True since the flow is legal.
If labels are assigned as described, then $(Z;V\setminus Z)$ is a minimum $s$-$t$ cut. Since $f$ is a maximum flow, we know that every minimum cut is saturated. Therefore, 
the result of $v$ in (4),(5) and (6) is \True, and $v$ outputs \True.
If all nodes output \True then, by~(1),(2) and (3), $f$ is a legal flow. By~(4) and (5), 
the $z(v)$ bits indicate an $s$-$t$ cut, and by~(6) cut $(Z;V\setminus Z)$ is saturated, and therefore, $f$ is a maximum flow, $\mf(G_s)=\True$.
\QED

For $\log f_{\max}\in \Omega(\log n)$, \lemmaref{Flow upper} gives a tight upper bound for $(\cF_{st},\mf)$ which concludes the proof of \theoremref{flow verification}.

Consider now the $\mfk$ problem as defined in \cite{KKP} over the 
family of 
configurations $\cF_{st}$.

\begin{definition}[$k$-Maximum Flow (\mfk)]\hfill\\
	\textbf{Instance}: A configuration $G_s\in \cF_{st}$.\\
	\textbf{Question}: Is the maximum flow between $s$ and $t$ in $G_s$ is exactly $k$?
\end{definition}

We give an upper bound for
$(\cF_{st},\mfk)$ in the  $\xc(r)$ model, which generalizes and 
improves the previous bound.

\begin{theorem}\label{mfk upper}
	For every $1\le r < 2\alpha$, there exists a PLS  for $(\cF_{st},\mfk)$
	in the $\xc(r)$ model, with verification complexity $O\left(\frac{\min\set{\alpha,k}}{r}(\log k+\log \Delta)\right)$, and for every $2\alpha\le r \le \Delta$, there exists a PLS  for 
	$(\cF_{st},\mfk)$ in the $\xc(r)$ model, with verification complexity $O(\log k)$.
\end{theorem}

\Proof
In a verification scheme for $(\cF_{st},\mfk)$, the prover can assign 
the flow values $f(v,u)$ for every 
edge $(v,u)$. W.l.o.g, assume that $f$ 
does not contain 
cycles of positive flow. In this case, 
$f_{\max}\le k$ and, since the flow value over each edge is an integer, 
the number of incident edges of every node $v$ carrying non-zero
flow is at most $2k$. By \lemmaref{Flow upper}, and the observation 
that it is sufficient that every node verifies the value of flow only 
on edges with $f(v,u)\neq 0$, the upper bounds follow.
\QED

To be precise, the problem solved in \cite{KKP} required in addition 
that every node holds the value $k$ in its state. Verifying that all 
nodes hold the same value $k$ is simply an additive $\log k$ factor to 
message length -- every node sends its value and verifies that all its 
neighbors have the same value. We argue in the following lemma, that 
$\Omega(\log k)$ is a lower bound for $(\cF_{st},\mfk)$ verification 
even if $k$ is known to all nodes. 

\begin{lemma}\label{mfk uni lower}
	For every $1\le k \le 2^{\Theta(n)}$, the verification complexity 
	of any PLS for $(\cF_{st},\mfk)$ is $\Omega(\log k)$, even in the 
	unicast model and for constant degree graphs.
\end{lemma}
\Proof
Consider the following graph family $\cF'$.
\begin{figure}
	\centering
	\vspace{-4mm}\includegraphics[width=3in]{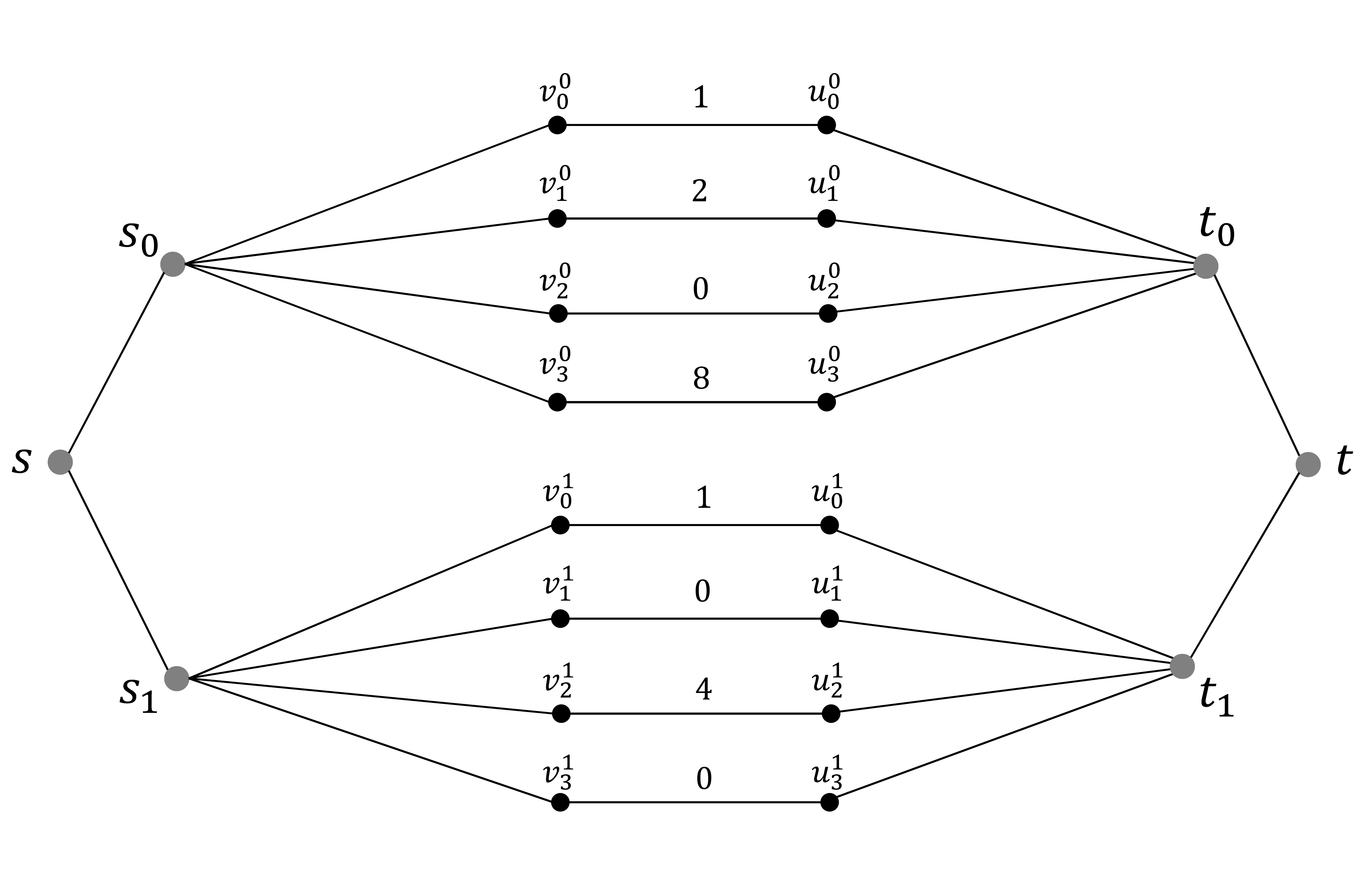}
	\caption{\it An example of $G_{a,b}$ configuration. All capacities which are not specified are $2^n$. In this example, the sum of capacities in the upper part is $11$ and in the lower part is $5$. Therefore, this is $G_{11,5}$.}
	\label{fig-flow}
\end{figure}
Node $s$ is connected to two nodes $s_0$ and $s_1$, and node $t$ is connected to two nodes $t_0$ and $t_1$ with edges of capacity $2^n$. Each of the nodes $s_0,s_1,t_0$ and $t_1$ is connected to $y$ additional different nodes with edges of capacity $2^n$. For $i\in\set{0,1}$ consider the following structure. Let $V^i=\set{v^i_0,\ldots,v^i_{y-1}}$ be the $y$ nodes connected to $s_i$ in addition to $s$ and let $U^i=\set{u^i_0,\ldots,u^i_{y-1}}$ be the $y$ nodes connected to $t_i$ in addition to $t$.
For every $0\le j \le y-1$ there is an edge $(v^i_j,u^i_j)$ with capacity either $c(v^i_j,u^i_j)=0$ or $c(v^i_j,u^i_j)=2^j$. In particular, $\sum_{j=0}^{y-1} c(v^i_j,u^i_j)$ can be every integer between $0$ and $2^y-1$.
We denote by $G_{a,b}$ the configuration where the sum of capacities in the subgraph induced by the set of nodes $V^0 \cup U^0$ is $a$, and the sum of capacities in the subgraph induced by the set of nodes $V^1 \cup U^1$ is $b$ (see \figref{fig-flow}). $\cF'=\set{G_{a,b}\mid 0\le a,b \le 2^y-1}$. Clearly, for every $k\le 2^y-1= 2^{\Omega(n)}$ and $0\le a \le k$ it holds that $G_{a,k-a}\in \cF'$ and $\mfk(G_{a,k-a})=\True$.

Assume by contradiction that there is a unicast proof-labeling scheme $\Sigma$ for the $(\cF_{st},\mfk)$ problem with verification complexity less than $\frac{\log k}{4}$. Consider the collection of $4$ messages sent over edges $(s,s_1)$ and $(t,t_1)$. By assumption, there are less than $\log k$ bits in this sequence of messages. Hence, there are less than $k$ different communication patterns over these edges. Therefore, there must be two configurations $G_{a,k-a}$ and $G_{a',k-a'}$, where $a\neq a'$, such that the communication pattern of $\Sigma$ over edges $(s,s_1)$ and $(t,t_1)$ is the same for both configurations. Consider the configuration $G_{a,k-a'}$. Obviously, since $a\neq a'$, $\mfk(G_{a,k-a'})=\False$. 

By construction, the state of every node $v\in W=\set{s,t,s_0,t_0}\cup V^0\cup U^0$ in $G_{a,k-a'}$ is the same as in $G_{a,k-a}$, and the state of every node $v\in W'=\set{s_1,t_1}\cup V^1\cup U^1$ in $G_{a,k-a'}$ is the same as in $G_{a',k-a'}$. 
Let $\ell_a(v)$ (respectively $\ell_{a'}(v)$) be the label assigned to node $v$ according to $\Sigma(G_{a,k-a})$ (respectively $\Sigma(G_{a',k-a'})$), and consider the following labeling $\ell$ for $G_{a,k-a'}$. 
For every $v\in W$ assign $\ell(v)=\ell_a(v)$, and for every $v\in W'$ assign $\ell(v)=\ell_{a'}(v)$.
Since the state and label of every $v\in W$ (respectively $v\in W'$) in $G_{a,k-a'}$ are exactly as in $\Sigma(G_{a,k-a})$ (respectively $\Sigma(G_{a',k-a'})$), all messages these nodes send are as in $\Sigma(G_{a,k-a})$ (respectively $\Sigma(G_{a',k-a'})$).

By assumption on the communication patterns of $\Sigma(G_{a,k-a})$ and $\Sigma(G_{a',k-a'})$, and the fact that $(s,s_1)$ and $(t,t_1)$ are the only edges in $W\times W'$, all nodes in $G_{a,k-a'}$ output \True, a contradiction to the correctness of $\Sigma$. Therefore, the verification complexity of any proof-labeling scheme for $(\cF',\mfk)$ is $\Omega(\log k)$. Since $\cF'\subset \cF_{st}$, the lower bound holds for $(\cF_{st},\mfk)$.

In order to show that this lower bound holds even for constant degree graphs, we change the construction so that every star structure induced by $\set{s_i}\cup V^i$, for $i\in\set{0,1}$, is replaced by a binary tree rooted at $s_i$ and its leaves are $V^i$. In the same way, we replace every star structure induced by $\set{t_i}\cup U^i$, for $i\in\set{0,1}$, by a binary tree. The maximum degree of the new graph family is $O(1)$, and the lemma follows.
\QED

By \theoremref{mfk upper}, this lower bound is tight for $2\alpha\le r \le \Delta$,  and the following theorem holds.

\begin{theorem}\label{mfk uni tight}
	For every $1\le k \le 2^{\Theta(n)}$ and every $2\alpha\le r \le 
	\Delta$, the verification complexity of $(\cF_{st},\mfk)$  in the 
	$\xc(r)$ model is $\Theta(\log k)$.
\end{theorem}

\section{Verification in Congested Cliques}\label{sec-clique}

In the congested clique model, the communication network is a fully connected graph
over $n$ nodes (i.e., an $n$-clique). Given an input graph $G=(V,E)$ with $n=|V|$, the nodes of $G$ are mapped
1-1 to the nodes of the clique, and the state of each node contains a bit
for each port, indicating whether the edge to that port is in $E$ or not, and,
if the edge is present and $G$ is weighted, the weight of the edge.
We assume that the part in the state that specifies whether the edge connected to this port is in $E$ is reliable:
since verification is done with respect to the given graph as input, there is no
way to verify its authenticity, but only whether the combination of input and output 
satisfies the given predicate. Moreover, we assume that the input is consistent, in the 
sense that the state at node $v$ indicates that $(v,u)$ is an edge in $E$ (possibly with some
weight $w$), if and only if so does the state of $u$ (namely edge agreement on the
input graph is guaranteed).

\subsection{Crossing  in Congested Cliques}

In what follows, we say that an edge is \emph{oriented} to indicate a specific order over its endpoints. 

\begin{definition}[Independent Edges]
	Let $G=(V,E)$ be a graph and let $e_1=(v_1,u_1)$ and $e_2=(v_2,u_2)$ be two oriented edges of $G$. The edges $e_1$ and $e_2$ are said to be \emph{independent} if and only if $v_1,u_1,v_2,u_2$ are four distinct nodes and $(v_1,u_2),(v_2,u_1)\notin E$.
\end{definition}

%
The following definition 
is illustrated in \figref{fig-cross}.
\begin{definition}[Crossing \cite{BFP}]
	\label{def-crossing}
	Let $G=(V,E)$ be a graph, let $e_1=(v_1,u_1)$ and 
	$e_2=(v_2,u_2)$ be two
	independent oriented edges of $G$, and for $i\in\set{1,2}$, let $p_i$ and $q_i$ be the port numbers of $e_i$ at $v_i$ and $u_i$ respectively.  The \emph{crossing} of $e_1$ 
	and $e_2$ in $G$, denoted by $G(e_1,e_2)$, is the graph obtained 
	from	$G$ by replacing $e_1$ and $e_2$ with the edges 
	$e'_1=(v_1,u_2)$ and $e'_2=(v_2,u_1)$ 
	so that $e'_1$ connects port $p_1$ at  $v_1$ and port $q_2$ at  $u_2$ and 
	$e'_2$ connects port $p_2$ at  $v_2$ and port $q_1$ at  $u_1$.
	%
\end{definition}

\begin{wrapfigure}{r}{.4\textwidth}
	\centering
	\vspace{-4mm}\includegraphics[width=.4\textwidth]{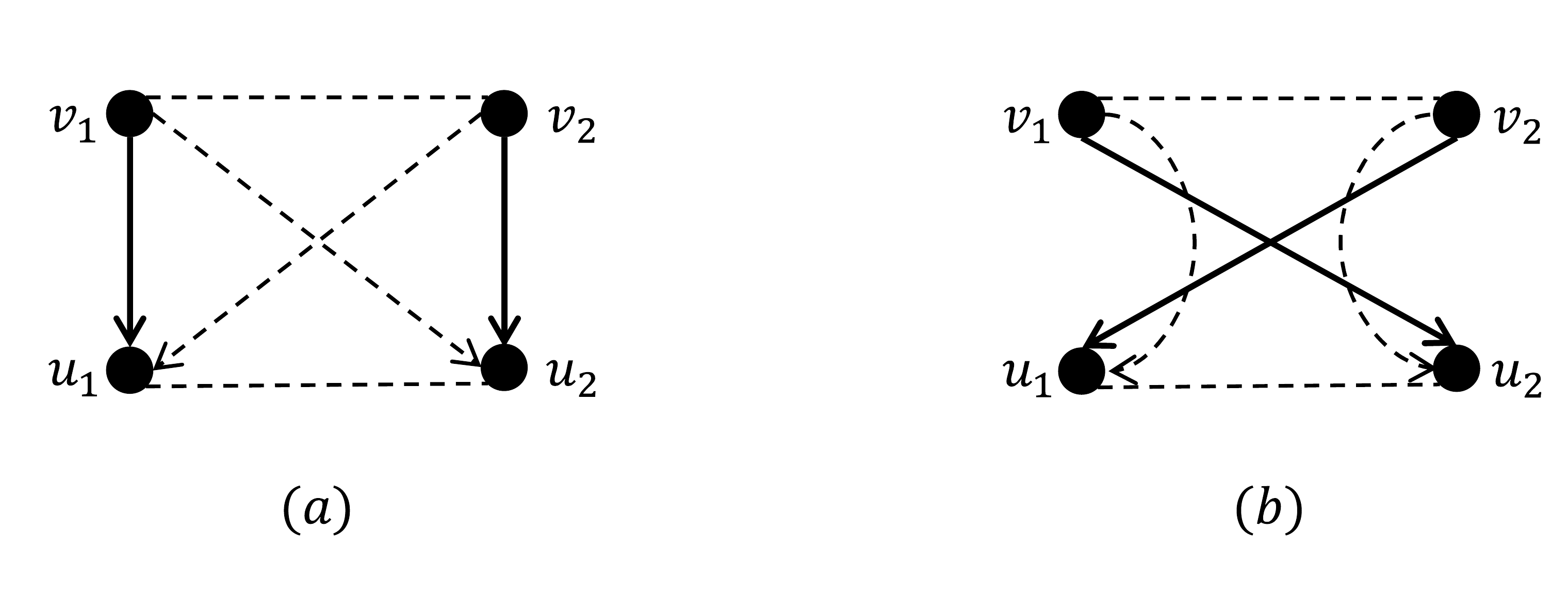}
	\caption{\it An illustration of the crossing operation on a clique 
		network. Solid edges are input graph edges, and dashed edged 
		are 
		communication-only edges. (a) edges ${e_1=(v_1,u_1)}$ and 
		$e_2=(v_2,u_2)$ are two
		independent oriented edges of an input graph $G$. (b) the 
		subgraph induced by nodes $v_1,u_1,v_2$ and $u_2$ in 
		$G(e_1,e_2)$.}
	\label{fig-cross}
\end{wrapfigure}

Consider an input graph $G=(V,E)$ in the clique, assume that
$e_1,e_2\in E$ are independent edges and let ${G(e_1,e_2)=(V,E')}$. 
Note that crossing a graph over a clique network does not result in a change of state:
Due to the port
preservation of the crossing operation, for every node $v\in V$ and every port 
$0\le i\le n-1$, the edge $(v,u)$ on port number $i$ in $G$ satisfies $(v,u)\in E$ if and only if the edge $(v,u')$ on port number $i$ in $G(e_1,e_2)$ satisfies $(v,u')\in E'$.

Whether we can prove a lower bound for verification in the congested clique for $r>1$ is still an open question. However, for the broadcast clique model (i.e., $r=1$), it turns out that we can.
The following lemma is the key to proving lower bounds for PLSs in the broadcast clique.
%
\begin{lemma}\label{Crossing in a BC Clique}
	Let $\cF$ be a family of configurations, let $\cP$ be a boolean
	predicate over $\cF$, and let $\Sigma$ be a PLS for
	$(\cF,\cP)$ in the  broadcast clique model
	with verification complexity $\kappa$.  Suppose that there
	is a configuration $G_s\in\cF$ such that $\cP(G_s)=\True$ and $G$ contains 
	$q$ pairwise independent oriented edges $e_1,\dots,e_q$.
	If $\kappa<\frac{\log q}{2}$, then there are $1\le i< j\le q$ such
	that $G_s(e_i,e_j)$ is accepted by $\Sigma$.
\end{lemma}
\Proof
Let $\Sigma=(\bP,\bV)$ be a PLS for
$(\cF,\cP)$ in the broadcast clique model, with verification complexity $\kappa$, and let $G_s$ be a configuration as described in the statement. Assume
that $\kappa<\frac{\log q}{2}$, and consider a collection of $q$ pairwise independent oriented edges $e_1=(v_1,u_1),\dots,e_q=(v_q,u_q)$. Let $\ell(v)$ is the label given by $\bP$ to $v$, let $M_v$ be the message sent by $v$ to all its neighbors according to $\bV$, and for every $i$, consider the bit-string $M^i=M_{v_i}\circ M_{u_i}$. We have $|M^i|<\log q$ for every $i$, and thus there are  less than $q$ possible distinct $M^i$'s in total. Therefore, by the pigeonhole principle, there are
$1\le i< j\le q$ such that $M^i=M^j$. Consider the output of the verifier $\bV$ in $G_s$ and in $G_s(e_i,e_j)$.  

By assumption, $G_s$ is accepted by $\Sigma$, i.e., with the labels provided by $\bP$,
the verifier $\bV$ outputs $\True$ at all nodes of~$G_s$. Therefore, clearly, all nodes other than  $v_i,u_i,v_j,u_j$  output $\True$ in $G_s(e_i,e_j)$. Now, consider node
$v_i$. Its neighbor $u_i$ in $G_s$ is replaced in
$G_s(e_i,e_j)$ by the node $u_j$, and its communication edge $(v_i,u_j)$ in $G_s$ is replaced in
$G_s(e_i,e_j)$ by communication edge $(v_i,u_i)$. Since $M_{u_i}=M_{u_j}$, the verifier acts the same at $v_i$ in
both $G_s$ and $G_s(e_i,e_j)$. The same argument works for $u_i,v_j$, and $u_j$, and therefore, the verifier also outputs
$\True$ at all nodes in $G_s(e_i,e_j)$, which implies that $G_s(e_i,e_j)$
is accepted by $\Sigma$.
%
\QED

We use the following corollary of \lemmaref{Crossing in a BC Clique}  to
lower-bound  verification complexity of broadcast clique PLSs.

\begin{corollary}\label{cor-crossing}
	Let $\cF$ be a family of configurations, and let $\cP$ be a boolean
	predicate over $\cF$. If there is a configuration $G_s\in\cF$
	satisfying that $\cP(G_s)=\True$ and $G$ contains 
	$q$ pairwise independent oriented edges $e_1,\dots,e_q$ such that for every $1\le i< j \le q$ it holds that $\cP(G_s(e_i,e_j))=\False$, then
	the verification complexity of any deterministic PLS 
	for $(\cF,\cP)$ in the broadcast clique model is $\Omega(\log q)$.
\end{corollary}

Note that we essentially cross two pairs of edges in the crossing operation: one pair of edges in $E$, and one pair of edges in $\bar{E}$. These two pairs are uniquely associated with each other in a way that if we assume a PLS in the $\xc(2)$ clique model, then we would not be able to apply the pigeonhole principle even with $1$-bit messages. To see why this is true, consider any set of independent oriented edges $(v_1,u_1),\dots,(v_q,u_q)$. For every $i\neq j$, both edges $(v_i,u_j),(v_j,u_i)\in \bar{E}$ are associated only with the pair of edges $(v_i,u_i),(v_j,u_j)\in E$. Therefore, with a PLS in the $\xc(2)$ clique model, it is possible that $M_{v_i}(u_j)\neq M_{v_j}(u_i)$ for every $i\neq j$ independently of other pairs. Hence, the crossing of any two edges may change the local view of at least one node. Therefore, the crossing technique can not be applied  for every $r>1$ in the congested clique.

\subsection{Minimum Spanning-Tree Verification}

In this section we illustrate the use of \corollaryref{cor-crossing} and prove tight bounds for the verification complexity of the Minimum Spanning-Tree (MST) problem.
Recall that an MST of a weighted graph $G$ is a spanning tree of $G$ whose sum of all its
edge-weights is minimum among all spanning trees of $G$. In particular, in the clique, there is a fully connected communication network, a weighted input graph $G=(V,E,w)$ where $E$ is a subset of communication edges, $w:E\rightarrow\mathbb{N}$ is the edge weight assignment, and a subset $T\subseteq E$ is specified as the MST. It is important to notice that all specifications of edge subsets are local in the sense that every node $v\in V$ has $n-1$ ports and in its state there is a specification for every edge $e_i$ on port number $i$ whether $e_i\in E$ and whether $e_i\in T$. According to our assumption on the clique model, the input graph $G$ is given in a reliable way, i.e., an edge $(v,u)$ is considered by $v$ to be in $E$ if and only if it is considered by $u$ to be in $E$. However, this consistency has to be verified for the edges of $T$.
In addition, since the communication network is fully connected and does not depend on the input graph $G$, we also consider the case where $G$ is disconnected. In this case, we define the MST as the set of minimum spanning-trees of all connected components of $G$.

Let $\cF_{{w_{\max}}}$ be the family of all weighted configurations 
(not necessarily connected) with maximum weight ${w_{\max}}$. Formally, 
if $e$ is an edge of the underlying weighted graph of
a configuration $G_s\in \cF_{w_{\max}}$, then ${w(e)\le w_{\max}}$. 
Edge weights are assumed to be known at their endpoints.

\begin{theorem}\label{MST verification}
	The verification complexity of $(\cF_{{w_{\max}}},\MST)$ in the 
	broadcast clique model is $\Theta(\log n + \log {w_{\max}})$.
\end{theorem}
\Proof
We first use \corollaryref{cor-crossing} to show $\Omega(\log n)$ lower bound for MST in the broadcast clique model. Consider the weighted graph $G=(V,E,w)$ where $V=\set{v_0,\ldots,v_{n-1}}$, $E=\set{(v_i,v_{i+1})\mid 0\le i\le n-2}$, and $w(e)=1$ for every $e\in E$. Intuitively, $G$ is a path of $n$ nodes with edges of weight $1$. We define the configuration $G_s$ to be the graph $G$ where $T=E$. Obviously, $\MST(G_s)=\True$. 
Next, we define the set of $q=\floor{\frac{n}{3}}-1$ independent oriented edges $e_1,\ldots,e_q$ as follows. For $1\leq i\leq \floor{\frac{n}{3}}-1$, let $e_i=(v_{3i},v_{3i+1})$. For every $1\leq i<j\leq \floor{\frac{n}{3}}-1$, 
$G'_s=G_s(e_i,e_j)$ is obtained from $G_s$ by removing edges
$(v_{3i},v_{3i+1})$ and $(v_{3j},v_{3j+1})$ from $G$, and
replacing them by $(v_{3i},v_{3j+1})$ and $(v_{3j},v_{3i+1})$.
Thus, the crossing
creates two connected components: the cycle $C=(v_{3i+1},v_{3i+2},\dots,v_{3j-1},u_{3j})$ and a path the contains the rest of the nodes, and therefore, $\MST(G'_s)=\False$. It follows from \corollaryref{cor-crossing}
that the verification complexity of any deterministic PLS for MST in the broadcast clique model is $\Omega(\log q)=\Omega(\log n)$.

\begin{figure}
	\centering
	\vspace{-4mm}\includegraphics[width=3in]{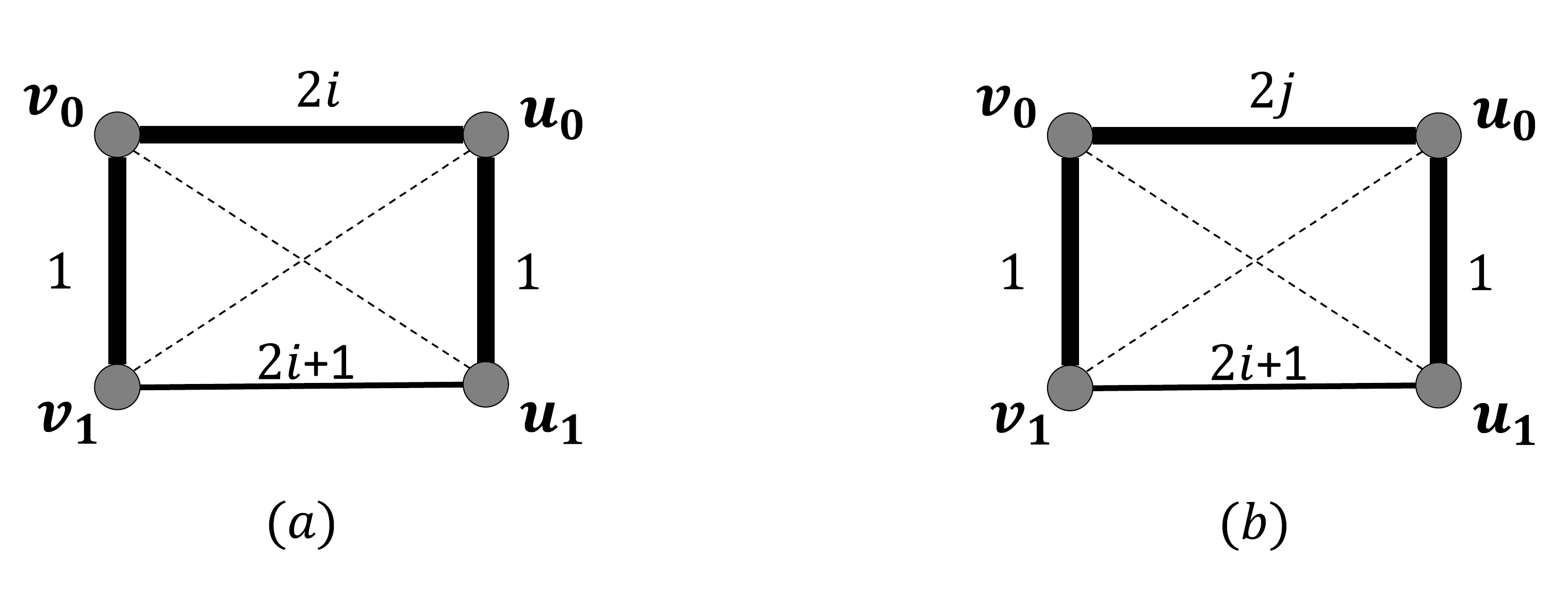}
	\caption{\it The configurations described in the proof of \theoremref{MST verification} for the $\Omega(\log W)$ lower bound. Dashed edged are communication edges, solid edges are in $E$ and thick solid edges are in $T$. (a) the configuration $G^i_s$ which satisfies $\MST(G^i_s)=\True$. (b) the configuration $G_s$ which satisfies $MST(G_s)=\False$ since $i<j$.}
	\label{fig-MST}
\end{figure}

For the $\Omega(\log W)$ lower bound, we show a variation of the proof in \cite{KKP}, which holds also for the broadcast clique model. Assume for contradiction that there
exists a  scheme $\Sigma$ for MST over $\cF_{W}$ in the broadcast clique model  with verification complexity $\kappa< \frac{1}{4}\log\left(\frac{W-1}{2}\right)$.
Let $G^i=(V,E,w_i)$ be a graph where $V=\set{v_0,u_0,v_1,u_1}$, $E=\set{(v_0,u_0), (v_1,u_1), (v_0,v_1), (u_0,u_1)}$, $w(v_0,v_1) = w(u_0,u_1) = 1$, $w(v_0,u_0)=2i$ and $w(v_1,u_1)=2i+1$ (see \figref{fig-MST}(a)). Let $G^i_s$ be the configuration over the graph $G^i$ with $T=E\setminus(v_1,u_1)$ . Obviously, $\MST(G^i_s)=\True$ for every $i\in\mathbb{N}$. In particular,
for every $1\le i \le \frac{W-1}{2}$ it holds that  $G^i_s\in \cF_{W}$ and $\Sigma$ accepts $G^i_s$. Let $\ell_i(v)$ be the label assigned by $\bP$ to $v$ in $G^i_s$.
Since $\kappa< \frac{1}{4}\log\left(\frac{W-1}{2}\right)$ and the fact that $\Sigma$ is a broadcast scheme, we get that
there exist $1\le i<j \le \frac{W-1}{2}$ such that $\vec{c}_\Sigma(G^i_s)=\vec{c}_\Sigma(G^j_s)$. Let $G_s$ be the same as $G^i_s$ except that $w(v_0,u_0)=2j$ (see \figref{fig-MST}(b)). Since $i<j$ it follows that $2i+1=w(v_1,u_1)<w(v_0,u_0)=2j$. Therefore, since $T=E\setminus(v_1,u_1)$, $\MST(G_s)=\False$. However, since $\vec{c}_\Sigma(G^i_s)=\vec{c}_\Sigma(G^j_s)$, with the labeling $\ell(v_0)=\ell_j(v_0)$, $\ell(u_0)=\ell_j(u_0)$, $\ell(v_1)=\ell_i(v_1)$ and $\ell(u_1)=\ell_i(u_1)$ for $G_s$ we
get that nodes $v_0$ and $u_0$ act exactly as in $G^j_s$ and output \True, and nodes $v_1$ and $u_1$ act exactly as in $G^i_s$ and output \True, a contradiction to the correctness of $\Sigma$.
This concludes the proof of the $\Omega(\log n + \log W)$ lower bound.

Finally, we show a  PLS for MST in the broadcast clique model with verification complexity $O(\log n + \log W)$. 
Consider the following scheme $(\bP,\bV)$. Given a legal configuration $G_s$, i.e., the set of edges $T$ is consistent and is an MST, the prover $\bP$ roots $T$, and gives every node a pointer to its parent in $T$. The verifier $\bV$ uses one communication round in which every non-root node sends its identity, the identity of its parent  and the weight of the edge connecting it to its parent; the
root sends an indication that it is the root. When all messages are received, each node locally constructs $T'$ from the collection of all edges sent by all nodes. Finally, every node $v$ outputs \True if the following conditions are met.	\begin{compactenum}
	\item\label{mst2} For all incident edges $e=(v,u)$: $e\in T$ if and only if $e\in T'$. 
	\item\label{mst1} $T'$ is a  tree spanning all $n$ nodes.
	\item\label{mst3} For all $e=(v,u)\in E$: if $e\notin T$ then $w(e)\ge w(e')$ for every $e'$ in the unique path between $v$ and $u$ in $T'$.
\end{compactenum}
We now prove the correctness of the scheme.
Recall that by the ``red rule'' (cf.~\cite{Tarjan:book}), the heaviest edge of every cycle is not in the MST.  Suppose $\MST(G_s)=\True$,  i.e., the set of edges $T$ is an MST. With the labels assigned by $\bP$, since $T$ is an MST and hence satisfies the red rule, all nodes output \True.
Assume now that all nodes output \True. 
Then by (\ref{mst2}), $T$ must be consistent over all nodes, and by (\ref{mst1}), we know that $T$ is a spanning tree. If (\ref{mst3}) holds at all nodes, $T$ satisfies the red rule and therefore $T$ is an MST, i.e., $\MST(G_s)=\True$.
\QED

\section{Randomized Proof-Labeling Schemes in the $\xc$ model}
\label{sec-RPLS}


The concept of randomized proof labeling schemes was introduced in 
\cite{BFP}. Briefly, the idea is that the messages generated by the 
verifier may depend not only on the local state and label, but also
on local random bits, and the correctness requirement is that if
$G_s$ satisfies $\cP$, then, using the labels assigned by the prover, 
all local verifiers accepts $G_s$; and 
if $\cP(G_s)=\False$ then, for every label assignment,
with probability at least $1/2$, at least one local verifier rejects 
$G_s$. (We consider only one-sided error RPLSs here.)

An RPLS for a given family $\cF$ of 
configurations and a boolean predicate $\cP$ over $\cF$, 
in the $\xc(r)$ model is defined in the same way deterministic PLS for
$(\cF,\cP)$ is defined, with the exception that the messages sent by
the verifier are a function of the local state, local label and a
string of random bits. The restriction of the $\xc(r)$ model means
that  the number of distinct messages (which are now random variables)
that may be sent out by a node is at most $r$. 
In accordance with our concern about dynamic partitioning of the
neighbors, we stress that in this case too, we assume that the
partitioning of neighbors into same-message groups is done
\emph{obliviously} of the actual random bits.
The correctness requirement are the same as in the standard
requirements from a (one-sided) RPLS:
%
A randomized scheme
$(\bP,\bV)$ for $(\cF,\cP)$ must satisfy the following
requirements, for every $G_s\in\cF$:

\begin{compactitem}
	\item If $\cP(G_s)=\True$ then, using the labels assigned by $\bP$, 
	the verifier $\bV$ accepts $G_s$.
	\item If $\cP(G_s)=\False$ then, for every label assignment,
	$\Pr[\bV\;\text{rejects}\; G_s]\geq \frac{1}{2}$.
\end{compactitem}

In this section, we extend the exponential relation between verification
complexity of deterministic and randomized schemes (shown in 
\cite{BFP} for broadcast deterministic schemes and unicast randomized  
schemes) to the $\xc(r)$ 
model. 

\begin{theorem}
	\label{thm-rpls}
	Let $\cF$ be a family of configurations, 
	let $\cP$ be a boolean predicate over $\cF$, and consider schemes 
	for $(\cF,\cP)$
	in the $\xc(r)$ model.
	If there exists 
	a (deterministic) PLS with verification complexity $\kappa_d$ then 
	there exists an RPLS with verification complexity $O(\log\kappa_d)$,
	and if there exists an RPLS with verification complexity $\kappa_r$ 
	then 
	there exists 
	a PLS with verification complexity $O(2^{\kappa_r})$.
\end{theorem}

\begin{proof}
	In \cite{BFP} it is shown that
	an RPLS (recall that this means a unicast RPLS) with verification
	complexity $O(\log \kappa)$ can be 
	constructed from a (broadcast) PLS with verification complexity $\kappa$.  The
	proof of this result can be easily adapted to show the following
	generalization. 
	%
	\begin{lemma}\label{rpls from upls}
		Let $\cF$ be a family of configurations
		and let $\cP$ be a boolean predicate over $\cF$.  If there exists
		a PLS for $(\cF,\cP)$ in 
		the $\xc(r)$ model with verification complexity $\kappa$, then
		there exists an RPLS for $(\cF,\cP)$ in the $\xc(r)$ model with verification
		complexity $O(\log \kappa)$.
	\end{lemma}
	The converse holds as well, as stated in  the following lemma.
	\begin{lemma}\label{upls from rpls}
		Let $\cF$ be a family of configurations
		and let $\cP$ be a boolean predicate over $\cF$.  If there exists
		a one-sided RPLS for $(\cF,\cP)$ in 
		the $\xc(r)$ model with verification complexity $\kappa$, then
		there exists a PLS for $(\cF,\cP)$ in the $\xc(r)$ model with verification
		complexity $O(2^\kappa)$.
	\end{lemma}
	%
	
	\Proof
	Let $(\bP,\bV)$ be a one-sided randomized proof-labeling scheme
	for $(\cF,\cP)$ in the $\xc(r)$ model with verification complexity $\kappa$. Let $G_s\in
	\cF$ be a configuration satisfying predicate $\cP$. For every node
	$v$, let $\ell(v)$ be the label assigned to $v$ by $\bP$, let
	$u_1,\ldots u_d$ be the $d=deg(v)$ neighbors of $v$ and let
	$g_1,\ldots,g_r$ be the partition of $u_1,\ldots u_d$ to $r$ groups
	used by $(\bP,\bV)$
	(some groups may be empty). For every non-empty
	group of neighbors $g_i$, let
	$C(v,g_i)=\Set{c^{(v,g_i)}_1,\ldots,c^{(v,g_i)}_y}$ be the collection of
	all certificates with positive probability to be sent from $v$ to its
	neighbors in group $g_i$ according to $\bV$ where labels are assigned
	by $\bP$.  By definition of $\xc(r)$, all neighbors that
	belong to the same group receive the same certificate. 
	We construct a deterministic
	proof-labeling scheme $(\bP',\bV')$ for $(\cF,\cP)$ in the $\xc(r)$
	model as
	follows. For every $v$, the label assigned by $\bP'$ to $v$ is
	$\ell'(v)=\ell(v)$. For every neighbor $u$ of $v$, let $g_v(u)$ be the
	group  containing $u$ in the partition of $v$. The message sent from $v$ to $u$
	according to $\bV'$ is
	$M'_v(u)=x^{(v,u)}_1,\ldots,x^{(v,u)}_{2^\kappa}$ where
	$x^{(v,u)}_j$ is a bit whose value is $1$ iff $j\in C(v,g_v(u))$. 
	Upon receiving $M'_{u_1}(v),\ldots,M'_{u_d}(v)$, the verifier at $v$ 
	outputs $\False$ if
	and only if there exist $j_1,\ldots,j_{d}$ such that $\bV(v)=\False$
	upon receiving $j_1,\ldots,j_d$
	from neighbors $u_1,\ldots u_{d}$ 
	respectively, and for every $1\leq i\leq d$ it holds that
	$x^{(u_i,v)}_{j_i}=1$. Intuitively, if for some combination of
	certificates in the support of $\bV$ it holds that $\bV(v)=\False$
	then $\bV'(v)=\False$, otherwise $\bV'(v)=\True$. This concludes the
	construction of $(\bP',\bV')$.
	
	The correctness of the scheme $(\bP',\bV')$ follows from the
	observation that 
	by construction, every combination of messages
	$j_1,\ldots,j_d$, such that for every $1\leq i\leq d$ it holds that
	$x^{(u_i,v)}_{j_i}=1$, has positive
	probability to occur in $\bV$. Therefore, since $(\bP,\bV)$ is a one-sided scheme, all
	combinations must lead to an output of $\True$. 
	Regarding communication complexity, the length of the messages sent by
	$\bV'$ is exactly the number of possible messages sent by $\bV$. Since
	the verification complexity of $(\bP,\bV)$ is $\kappa$, the number of
	bits in the messages of $\bV'$ is $2^\kappa$.
\end{proof}

\section{Conclusion}\label{sec-conc}
In this paper we studied the $\xc(r)$ model from the perspective of
verification. This angle seems particularly convenient, because it 
involves a single round of message exchange. (If multiple rounds
are allowed, one has to consider the possibility of reconfiguring
the neighbor partitions: is it allowed to partition the neighbors
anew in each round, and if so, at what cost?).
We focus on the
relation between
the number of different messages of 
each node and the verification complexity of proof-labeling
schemes. We gave tight bounds on the verification complexity of 
edge agreement and max flow in the $\xc(r)$ model.
We have shown that in the restrictive broadcast
model, a well defined matching
is harder to verify than the maximality of a
given matching, and that it is possible to obtain lower bounds on the verification
complexity in congested cliques.
Many interesting
questions remain open. We list a few below.
\begin{compactitem}
	\item Develop a theory for a restricted number of interface cards (NICs). The
	number of NICs limits the number of messages that can be simultaneously
	transmitted. In this paper we looked only at a simple  case of one round of
	communication.
	We believe that developing a tractable and realistic model in which the
	number of NICs is a parameter is an important challenge.
	\item As mentioned, in multiple round algorithms,  dynamic 
	reconfigurations
	can be exploited to convey information. It seems
	that an interesting challenge would be to account for  
	dynamic reconfigurations.
	\item We considered a model in which a single parameter $r$ is used to indicate the restriction of all nodes. What can be said about a model in which every node has its own restriction?
	\item We have given examples of problems that have a linear improvement in verification complexity as a function $r$, and on the other hand, we have given examples of problems that are not sensitive at all to $r$. Can a characterization of problems be shown, according to their sensitivity of verification complexity to $r$?
\end{compactitem}

\bibliographystyle{abbrv}
\bibliography{UPLS}

\begin{thebibliography}{10}

\bibitem{AhujaMO-93}
R.~K. Ahuja, T.~L. Magnanti, and J.~B. Orlin.
\newblock {\em Network Flows}.
\newblock Prentice-Hall, Engelwood Cliffs, New Jersey, 1993.

\bibitem{AFIM}
H.~Arfaoui, P.~Fraigniaud, D.~Ilcinkas, and F.~Mathieu.
\newblock Distributedly testing cycle-freeness.
\newblock In {\em 40th Int. Workshop on Graph-Theoretic Concepts in Computer
  Science (WG)}, LNCS, pages 15--28. Springer, 2014.

\bibitem{AFP13}
H.~Arfaoui, P.~Fraigniaud, and A.~Pelc.
\newblock Local decision and verification with bounded-size outputs.
\newblock In {\em 15th Symp. on Stabilization, Safety, and Security of
  Distributed Systems (SSS)}, LNCS, pages 133--147. Springer, 2013.

\bibitem{APV91}
B.~Awerbuch, B.~Patt-Shamir, and G.~Varghese.
\newblock Self-stabilization by local checking and correction.
\newblock In {\em 32nd Symposium on Foundations of Computer Science (FOCS)},
  pages 268--277. IEEE, 1991.

\bibitem{BFP}
M.~Baruch, P.~Fraigniaud, and B.~Patt{-}Shamir.
\newblock Randomized proof-labeling schemes.
\newblock In {\em Proc.\ 34th {ACM} Symp.\ on Principles of Distributed
  Computing ({PODC})}, pages 315--324, 2015.

\bibitem{BeckerARR16}
F.~Becker, A.~{Fern{\'{a}}ndez Anta}, I.~Rapaport, and E.~R{\'{e}}mila.
\newblock The effect of range and bandwidth on the round complexity in the
  congested clique model.
\newblock In {\em Proc.\ 22nd Int.\ Conf.\ on Computing and Combinatorics
  ({COCOON})}, pages 182--193, 2016.

\bibitem{BFP14}
L.~Blin, P.~Fraigniaud, and B.~Patt-Shamir.
\newblock On proof-labeling schemes versus silent self-stabilizing algorithms.
\newblock In {\em 16th Int. Symp. on Stabilization, Safety, and Security of
  Distributed Systems (SSS)}, LNCS, pages 18--32. Springer, 2014.

\bibitem{DH+12}
A.~{Das Sarma}, S.~Holzer, L.~Kor, A.~Korman, D.~Nanongkai, G.~Pandurangan,
  D.~Peleg, and R.~Wattenhofer.
\newblock Distributed verification and hardness of distributed approximation.
\newblock {\em SIAM J. Comput.}, 41(5):1235--1265, 2012.

\bibitem{DruckerKO-14}
A.~Drucker, F.~Kuhn, and R.~Oshman.
\newblock On the power of the congested clique model.
\newblock In {\em Proceedings of the 2014 ACM Symposium on Principles of
  Distributed Computing}, PODC '14, pages 367--376, New York, NY, USA, 2014.
  ACM.

\bibitem{FFH16}
L.~Feuilloley, P.~Fraigniaud, and J.~Hirvonen.
\newblock {A Hierarchy of Local Decision}.
\newblock In {\em 43rd Int. Colloquium on Automata, Languages, and Programming
  (ICALP 2016)}, pages 118:1--118:15. Schloss Dagstuhl--Leibniz-Zentrum fuer
  Informatik, 2016.

\bibitem{FLSW16}
K.-T. Foerster, T.~Luedi, J.~Seidel, and R.~Wattenhofer.
\newblock Local checkability, no strings attached.
\newblock In {\em Proceedings of the 17th International Conference on
  Distributed Computing and Networking}, ICDCN '16, pages 21:1--21:10, New
  York, NY, USA, 2016. ACM.

\bibitem{FRSW17}
K.-T. Foerster, O.~Richter, J.~Seidel, and R.~Wattenhofer.
\newblock Local checkability in dynamic networks.
\newblock In {\em Proceedings of the 18th International Conference on
  Distributed Computing and Networking}, ICDCN '17, pages 4:1--4:10, New York,
  NY, USA, 2017. ACM.

\bibitem{FGKS13}
P.~Fraigniaud, M.~G\"{o}\"{o}s, A.~Korman, and J.~Suomela.
\newblock What can be decided locally without identifiers?
\newblock In {\em Proceedings of the 2013 ACM Symposium on Principles of
  Distributed Computing}, PODC '13, pages 157--165, New York, NY, USA, 2013.
  ACM.

\bibitem{FHK12}
P.~Fraigniaud, M.~M. Halld{\'o}rsson, and A.~Korman.
\newblock On the impact of identifiers on local decision.
\newblock In {\em Principles of Distributed Systems: 16th International
  Conference, OPODIS. Proceedings}, pages 224--238. Springer Berlin Heidelberg,
  2012.

\bibitem{FHS15}
P.~Fraigniaud, J.~Hirvonen, and J.~Suomela.
\newblock Node labels in local decision.
\newblock In {\em Structural Information and Communication Complexity: 22nd
  International Colloquium, SIROCCO. Post-Proceedings}, pages 31--45. Springer
  International Publishing, 2015.

\bibitem{FKP13}
P.~Fraigniaud, A.~Korman, and D.~Peleg.
\newblock Towards a complexity theory for local distributed computing.
\newblock {\em J. ACM}, 60(5):35, 2013.

\bibitem{FRT13}
P.~Fraigniaud, S.~Rajsbaum, and C.~Travers.
\newblock Locality and checkability in wait-free computing.
\newblock {\em Distributed Computing}, 26(4):223--242, 2013.

\bibitem{FRT14}
P.~Fraigniaud, S.~Rajsbaum, and C.~Travers.
\newblock On the number of opinions needed for fault-tolerant run-time
  monitoring in distributed systems.
\newblock In {\em 5th Int.\ Conf.\ on Runtime Verification}, LNCS, pages
  92--107. Springer, 2014.

\bibitem{GS11}
M.~G{\"o}{\"o}s and J.~Suomela.
\newblock Locally checkable proofs.
\newblock In {\em 30th ACM Symp. on Principles of Distributed Computing
  (PODC)}, pages 159--168, 2011.

\bibitem{KK07}
A.~Korman and S.~Kutten.
\newblock Distributed verification of minimum spanning trees.
\newblock {\em Distributed Computing}, 20:253--266, 2007.

\bibitem{KKM}
A.~Korman, S.~Kutten, and T.~Masuzawa.
\newblock Fast and compact self stabilizing verification, computation, and
  fault detection of an {MST}.
\newblock In {\em 30th Ann.\ {ACM} Symposium on Principles of Distributed
  Computing (PODC)}, pages 311--320, 2011.

\bibitem{KKP}
A.~Korman, S.~Kutten, and D.~Peleg.
\newblock Proof labeling schemes.
\newblock {\em Distributed Computing}, 22(4):215--233, 2010.

\bibitem{KN}
E.~Kushilevitz and N.~Nisan.
\newblock {\em Communication complexity}.
\newblock Cambridge University Press, 1997.

\bibitem{MB83}
D.~W. Matula and L.~L. Beck.
\newblock Smallest-last ordering and clustering and graph coloring algorithms.
\newblock {\em J. ACM}, 30(3):417--427, July 1983.

\bibitem{NW61}
C.~S.~A. Nash-Williams.
\newblock Edge-disjoint spanning trees of finite graphs.
\newblock {\em Journal of the London Mathematical Society}, s1-36(1):445--450,
  1961.

\bibitem{NW64}
C.~S.~A. Nash-Williams.
\newblock Decomposition of finite graphs into forests.
\newblock {\em Journal of the London Mathematical Society}, s1-39(1):12, 1964.

\bibitem{peleg:book}
D.~Peleg.
\newblock {\em Distributed Computing: A Locality-Sensitive Approach}.
\newblock Society for Industrial and Applied Mathematics, Philadelphia, PA,
  USA, 2000.

\bibitem{Tarjan:book}
R.~E. Tarjan.
\newblock {\em Data Structures and Network Algorithms}, volume~44 of {\em
  CBMS}.
\newblock Society for Industrial and Applied Mathematics, 1983.

\end{thebibliography}

\end{document}